\RequirePackage{fix-cm}
\documentclass[twocolumn,natbib]{svjour3}          
\smartqed  
\usepackage{graphicx}

\usepackage[english,british]{babel}

\usepackage{amsthm}
\usepackage{amsfonts}
\usepackage{amsmath}

\usepackage{algorithm}
\usepackage{algpseudocode}

\usepackage{verbatim}
\usepackage{paralist}

\usepackage[style=base]{caption}
\usepackage{subfig}

  \theoremstyle{plain}
  \newtheorem{prop}{\protect\propositionname}
  \theoremstyle{plain}

\journalname{Statistics and Computing}
\begin{document}

\title{Bayesian model comparison with un-normalised likelihoods
}


\author{Richard G. Everitt        \and
        Adam M. Johansen \and
        Ellen Rowing \and
        Melina Evdemon-Hogan
}


\institute{R. G. Everitt \and E. Rowing \and M. Evdemon-Hogan \at
              Department of Mathematics and Statistics, University of Reading, UK. \\
              \email{r.g.everitt@reading.ac.uk}           
           \and
           A. M. Johansen \at
              Department of Statistics, University of Warwick, Coventry, CV4 7AL, UK. \\
              \email{a.m.johansen@warwick.ac.uk}           
}

\date{Received: date / Accepted: date}

\maketitle

\begin{abstract}
Models for which the likelihood function can be evaluated only up to a
parameter-dependent unknown normalizing constant, such as Markov random field
models, are used widely in computer science, statistical physics,  spatial
statistics, and network analysis. However, Bayesian analysis of these models
using standard Monte Carlo methods is not possible due to the intractability
of their likelihood functions. Several methods that permit exact, or close to
exact, simulation from the posterior distribution have recently been
developed. However, estimating the evidence and Bayes' factors (BFs) for these
models remains challenging in general. This paper describes new random weight
importance sampling and sequential Monte Carlo methods for estimating BFs that
use simulation to circumvent the evaluation of the intractable likelihood, and
compares them to existing methods. In some cases we observe an advantage in
the use of \textit{biased} weight estimates. An initial investigation into the
theoretical and empirical properties of this class of methods is
presented. Some support for the use of biased estimates is presented, but we
advocate caution in the use of such estimates.

\end{abstract}

\section{Introduction\label{sec:Introduction}}

There has been much recent interest in performing \linebreak Bayesian inference
in models where the posterior is intractable and, in particular, we have
the situation in which the posterior distribution $\pi(\theta|y)\propto p(\theta)f(y|\theta)$,
cannot be evaluated pointwise. This intractability typically occurs
occurs due to the intractability of the likelihood, i.e. $f(y|\theta)$
cannot be evaluated pointwise. Example scenarios include:
\begin{enumerate}
\item the use of big data sets, where $f(y|\theta)$ consists of a product
of a large number of terms; 
\item the existence of a large number of latent variables $x$, with
$f(y|\theta)$ known only as a high dimensional integral $f(y|\theta)=\int_{x}f(y,x|\theta)dx$; 
\item when $f(y|\theta)=\frac{1}{Z(\theta)}\gamma(y|\theta)$, with $Z(\theta)$
being an intractable normalising constant (INC) for the tractable
term $\gamma(y|\theta)$ (e.g. when $f$ factorises as a Markov random
field); 
\item where it is possible to sample from $f(\cdot|\theta)$, but not to
evaluate it, such as when the distribution of the data given $\theta$
is modelled by a complex stochastic computer model. 
\end{enumerate}
Each of these (overlapping) situations has been considered in some
detail in previous work and each has inspired different methodologies.

In this paper we focus on the third case, in which the likelihood
has an INC. This is an important problem in its own right (\citet{Girolami2013}
refer to it as ``one of the main challenges to methodology for computational
statistics currently''). There exist several competing methodologies
for inference in this setting (see \citet{Everitt2012}). In particular,
the \emph{``exact'' approaches} of \citet{Moller2006} and \citet{Murray2006}
exploit the decomposition $f(y|\theta)=\frac{1}{Z(\theta)}\gamma(y|\theta)$,
whereas \emph{``simulation based'' methods} such as approximate
Bayesian computation (ABC) \citep{Grelaud2009} do not depend upon
such a decomposition and can be applied more generally: to situation
1 in \citet{Picchini2013}; situations 2 and 3 (e.g. \citet{Everitt2012})
and situation 4 (e.g. \citet{Wilkinson2008}).

This paper considers the problem of Bayesian model comparison in the
presence of an INC. We explore both exact and simulation-based methods,
and find that elements of both approaches may also be more generally
applicable. Specifically: 
\begin{compactitem}
\item For exact methods we find that approximations are required to allow
practical implementation, and this leads us to investigate the use
of approximate weights in importance sampling (IS) and sequential
Monte Carlo (SMC). We examine the use of both \emph{exact-approximate}
approaches (as in \citet{Fearnhead2010f}) and also ``\emph{inexact-approximate}'' methods, in which complete
flexibility is allowed in the approximation of weights, at the cost
of losing the exactness of the method. This work is a natural counterpart
to \citet{Alquier2014}, which examines ananalogous question
(concerning the acceptance probability) for Markov chain Monte Carlo
(MCMC) algorithms. These generally applicable methods, ``noisy MCMC''
\citep{Alquier2014} and ``noisy SMC'' (this paper) have some potential
to address situations 1-3. 
\item We provide some comparison of these inexact approximations with
  simulation-based methods, including the ``synthetic likelihood'' (SL) of
  \citet{Wood2010f}. In the applications considered here we find this to be a
  viable alternative to ABC. Our results are suggestive that this, and related
  methods, may find success in scenarios in which ABC is more usually applied. 
\end{compactitem}
In the remainder of this section we briefly outline the problem of, and
methods for, parameter inference in the presence of an INC.
We then detail the problem of Bayesian model comparison in this context,
before discussing methods for addressing it in the following two sections.

\subsection{Parameter inference\label{sub:Parameter-inference}}

In this section we consider the problem of simulating from $\pi(\theta|y)\propto p(\theta)\gamma(y|\theta)/Z(\theta)$
using MCMC. This problem has been well studied, and such models are
termed\emph{ doubly intractable} because the acceptance probability
in the Metropolis-Hastings (MH) algorithm 
\begin{equation}
\min\left\{ 1,\frac{q(\theta|\theta^{*})}{q(\theta^{*}|\theta)}\frac{p(\theta^{*})}{p(\theta)}\frac{\gamma(y|\theta^{*})}{\gamma(y|\theta)}\frac{Z(\theta)}{Z(\theta^{*})}\right\} ,\label{eq:intractable_accept}
\end{equation}
cannot be evaluated due to the presence of the INC. We first review
exact methods for simulating from such a target in sections \ref{sub:Single-and-multiple}-\ref{sub:Russian-Roulette}, before
looking at simulation-based methods in sections \ref{sub:Approximate-Bayesian-computation}
and \ref{sub:Synthetic-likelihood}. The methods described here in
the context of MCMC form the basis of the methods for evidence estimation
developed in the remainder of the paper.

\subsubsection{Single and multiple auxiliary variable methods\label{sub:Single-and-multiple}}

\selectlanguage{english}%
\citet{Moller2006} avoid the evaluation of the INC by augmenting
the target distribution with an extra variable $u$ that lies on the
same space as $y$, and use an MH algorithm with target distribution
\begin{equation*}
\pi(\theta,u|y)\propto q_{u}(u|\theta,y)f(y|\theta)p(\theta),
\end{equation*}
where $q_{u}$ is some (normalised) arbitrary distribution. As the
MH proposal in $(\theta,u)$-space they use 
\begin{equation*}
(\theta^{*},u^{*})\sim f(u^{*}|\theta^{*})q(\theta^{*}|\theta),
\end{equation*}
giving an acceptance probability of 
\begin{equation*}
\min\left\{ 1,\frac{q(\theta|\theta^{*})}{q(\theta^{*}|\theta)}\frac{p(\theta^{*})}{p(\theta)}\frac{\gamma(y|\theta^{*})}{\gamma(y|\theta)}\frac{q_{u}(u^{*}|\theta^{*},y)}{\gamma(u^{*}|\theta^{*})}\frac{\gamma(u|\theta)}{q_{u}(u|\theta,y)}\right\} .
\end{equation*}

Note that, by viewing $q_{u}(u^{*}|\theta^{*},y)/\gamma(u^{*}|\theta^{*})$
as an unbiased IS estimator of $1/Z(\theta^{*})$, this algorithm
can be seen as an instance of the \emph{exact approximations} described
in \citet{Beaumont2003} and \citet{Andrieu2009}, where it is established
that if an unbiased estimator of a target density is used appropriately
in an MH algorithm, the $\theta$-marginal of the invariant distribution
of this chain is the target distribution of interest. This automatically
suggests extensions to the \emph{single auxiliary variable (SAV)}
method described above, where $M>1$ importance points are used, yielding:
\begin{equation}
\widehat{\frac{1}{Z(\theta)}}=\frac{1}{M}\sum_{m=1}^{M}\frac{q_{u}(u^{(m)}|\theta,y)}{\gamma(u^{(m)}|\theta)}.\label{eq:1overz_sav_multiplepoints}
\end{equation}
\citet{Andrieu2012d} show that the reduced variance of this estimator
leads to a reduced asymptotic variance of estimators from the resultant
Markov chain. The variance of the IS estimator is strongly dependent
on an appropriate choice of IS target $q_{u}(\cdot|\theta,y)$, which
should have lighter tails than $f(\cdot|\theta)$. \citet{Moller2006}
suggest that a reasonable choice may be $q_{u}(\cdot|\theta,y)=f(\cdot|\widehat{\theta})$,
where $\widehat{\theta}$ is the maximum likelihood estimator of $\theta$.
However, in practice $q_{u}(\cdot|\theta,y)$ can be difficult to
choose well, particularly when $y$ lies on a high dimensional space.
Motivated by this, annealed importance sampling (AIS) \citep{Neal2001}
can be used as an alternative to IS, leading to the \emph{multiple auxiliary
  variable (MAV)} method of \citet{Murray2006}. AIS makes use of a sequence of
$K$ targets, which in \citet{Murray2006} are chosen to be 
\begin{eqnarray}
f_{k}(\cdot|\theta,\widehat{\theta},y) & \propto & \gamma_{k}(\cdot|\theta,\widehat{\theta},y) \nonumber\\ & = & \gamma(\cdot|\theta){}^{(K+1-k)/(K+1)}q_{u}(\cdot|\theta,y){}^{k/(K+1)} \label{eq:ais_seq}
\end{eqnarray}
between $f(\cdot|\theta)$ and $q_{u}(\cdot|\theta,y)$. After the
initial draw $u_{K+1}\sim f(\cdot|\theta)$, the auxiliary point is
taken through a sequence of $K$ MCMC moves which successively have
target \linebreak $f_{k}(\cdot|\theta,\widehat{\theta},y)$ for $k=K:1$. The
resultant IS estimator is given by 
\begin{equation}
\widehat{\frac{1}{Z(\theta)}}=\frac{1}{M}\sum_{m=1}^{M}\prod_{k=1}^{K}\frac{\gamma_{k}(u_{k-1}^{(m)}|\theta,\widehat{\theta},y)}{\gamma_{k-1}(u_{k-1}^{(m)}|\theta,\widehat{\theta},y)}.\label{eq:1overz_mav_multiplepoints}
\end{equation}
This estimator has a lower variance (although at a higher computational
cost) than the corresponding IS estimator. We note that AIS can be viewed as a
particular case of SMC without resampling and one might expect to obtain
additional improvements at negligible cost by incorporating resampling steps
within such algorithms (see \citet{Zhou2013} for an illustration of the potential
improvement and some discussion); we do not pursue this here as it is not the
focus of this work.

\subsubsection{Exchange algorithms\label{sub:Exchange-algorithms}}

An alternative approach to avoiding the ratio of INCs in equation
\eqref{eq:intractable_accept} is given by \citet{Murray2006},
in which it is suggested to use the acceptance probability
\[
\min\left\{ 1,\frac{q(\theta|\theta^{*})}{q(\theta^{*}|\theta)}\frac{p(\theta^{*})}{p(\theta)}\frac{\gamma(y|\theta^{*})}{\gamma(y|\theta)}\frac{\gamma(u|\theta)}{\gamma(u|\theta^{*})}\right\} ,
\]
where $u\sim f(\cdot|\theta^{*})$, motivated by the intuitive idea
that $\gamma(u|\theta)/\gamma(u|\theta^{*})$ is a single point IS
estimator of \linebreak $Z(\theta)/Z(\theta^{*})$. This method is shown to have
the correct invariant distribution, as is the extension in which AIS
is used in place of IS. A potential extension might seem to be using
multiple importance points $\{u^{(m)}\}_{m=1}^{M}\sim f(\cdot|\theta^{*})$
to obtain an estimator of $Z(\theta)/Z(\theta^{*})$ that has a smaller
variance, with the aim of improving the statistical efficiency of
estimators based on the resultant Markov chain. This scheme is shown
to work well empirically in \citet{Alquier2014}. However, this chain
does not have the desired target as its invariant distribution. Instead
it can be seen as part of a wider class of algorithms that use a noisy
estimate of the acceptance probability: \emph{noisy Monte Carlo} algorithms
(also referred to as \emph{``inexact approximations''} in \citet{Girolami2013}).
\citet{Alquier2014} shows that under uniform ergodicity of the ideal
chain, a bound on the expected difference between the noisy and true
acceptance probabilities can lead to bounds on the distance between
the desired target distribution and the iterated noisy kernel. It
also describes additional noisy MCMC algorithms for approximately
simulating from the posterior, based on Langevin dynamics.

\subsubsection{Russian Roulette and other approaches\label{sub:Russian-Roulette}}
\citet{Girolami2013} \foreignlanguage{british}{use series-based approximations
to intractable target distributions within the \linebreak exact-approximation
framework, where ``Russian Roulette'' methods from the physics literature
are used to ensure the unbiasedness of truncations of infinite sums.
These methods do not require exact simulation from $f(\cdot|\theta^{*})$,
as do the SAV and exchange approaches described in the previous
two sections. However, SAV and exchange are often implemented in practice
by generating the auxiliary variables by taking the final point of
a long ``internal'' MCMC run in place of exact simulation (e.g \citet{Caimo2011}).
For finite runs of the internal MCMC, this approach will not have
exactly the desired invariant distribution, but \citet{Everitt2012}
shows that under regularity conditions the bias introduced by this
approximation tends to zero as the run length of the internal MCMC
increases: the same proof holds for the use of an MCMC chain for the
simulation within an ABC-MCMC (i.e. MCMC applied to an ABC approximation of
the posterior, \citet{Marjoram2003a}) or SL-MCMC (i.e. MCMC applied to an SL approximation) algorithm, as described in
sections \ref{sub:Approximate-Bayesian-computation} and \ref{sub:Synthetic-likelihood}.
Although the approach of \citet{Girolami2013} is exact, as they note
it is significantly more computationally expensive than this
approximate approach. For this reason, we do not pursue Russian Roulette
approaches further in this paper.}

When a rejection sampler is available for simulating from $f(\cdot|\theta^{*})$,
\citet{Rao2013} introduce an alternative exact algorithm that has
some favourable properties compared to the exchange algorithm. Since
a rejection sampler is not available in many cases, we do not pursue
this approach further.

\subsubsection{Approximate Bayesian computation\label{sub:Approximate-Bayesian-computation}}
Approximate Bayesian Computation \citep{Tavare1997f} refers to methods that aim to approximate
an intractable likelihood $f(y|\theta)$ through the integral 
\begin{equation}
  \widetilde{f}(S(y)|\theta) \propto \int \pi_\epsilon(S(y)|S(u)) f(u|\theta) du \label{eq:abc_integral}
\end{equation}
where $S(\cdot)$ gives a vector of summary statistics
and $\pi_{\epsilon}\left(S(y)|S(u)\right)$ is the density of a symmetric kernel with
bandwidth $\epsilon$, centered at $S(u)$ and evaluated at $S(y)$. As $\epsilon\rightarrow0$,
this distribution becomes more concentrated, so that
in the case where $S(\cdot)$ gives sufficient statistics for estimating
$\theta$, as $\epsilon\rightarrow0$ the approximate posterior becomes
closer to the true posterior. This approximation is used within standard
Monte Carlo methods for simulating from the posterior. For example,
it may be used within an MCMC algorithm,
where using an exact-approximation argument it can be seen that it
is sufficient in the calculation of the acceptance probability to
use the Monte Carlo approximation\foreignlanguage{english}{ 
\begin{equation}
\widehat{f}_{\epsilon}(S(y)|\theta^{*})=\frac{1}{M}\sum_{m=1}^{M}\pi_{\epsilon}\left(S(y) \left\vert S\left(u^{(m)}\right)\right.\right)\label{eq:abc_llhd}
\end{equation}
}for the likelihood at $\theta^{*}$ at each iteration, where \linebreak \foreignlanguage{english}{$\{u^{(m)}\}_{m=1}^{M}\sim f(\cdot|\theta^{*})$}.
Whilst the exact-approximation argument means that there is no additional
bias due to this Monte Carlo approximation, the approximation introduced
through using a tolerance $\epsilon>0$ or insufficient summary statistics
may be large. For this reason it might be considered a last resort
to use ABC on likelihoods with an INC, but previous success on these
models (e.g \citet{Grelaud2009} and \citet{Everitt2012}) lead us
to consider them further in this paper.

\subsubsection{Synthetic likelihood\label{sub:Synthetic-likelihood}}
ABC is essentially using, based on simulations from $f$, a nonparameteric
estimator of $f_{S}\left(S|\theta\right)$, the distribution of the
summary statistics of the data given $\theta$. In some situations,
a parametric model might be more appropriate. For example, if the statistic is
the sum of independent random variables, a Central Limit Theorem (CLT) might
imply that it would be appropriate to assume that $f_{S}\left(S|\theta\right)$ is a
multivariate Gaussian.  

The SL approach \citep{Wood2010f} proceeds by making exactly this
Gaussian assumption and uses this approximate likelihood within an
MCMC algorithm. The parameters (the mean and variance) of
this approximating distribution for a given $\theta$ are estimated
based on the summary statistics of simulations $\{u^{(m)}\}_{m=1}^{M}\sim f(\cdot|\theta)$.
Concretely, the estimate of the likelihood is
\begin{equation*}
\widehat{f}_{\mbox{SL}}\left(S(y)|\theta\right)=\mathcal{N}\left(S(y);\widehat{\mu}_{\theta},\widehat{\Sigma}_{\theta}\right),\label{eq:sl_llhd}
\end{equation*}
where 
\begin{align}
\widehat{\mu}_{\theta}=&\frac{1}{M}\sum_{m=1}^{M}S\left(u^{(m)}\right) & \widehat{\Sigma}_{\theta}&=\frac{ss^{T}}{M-1},\label{eq:sl_params}
\end{align}
with $s=\left(S\left(u_{1}\right)-\widehat{\mu}_{\theta},...,S\left(u_{M}\right)-\widehat{\mu}_{\theta}\right).$
\citet{Wood2010f} applies this method in a setting where the summary
statistics are regression coefficients, motivated by their distribution
being approximately normal. One of the approximations inherent in
this method, as in ABC, is the use of summary statistics rather than
the whole dataset. However, unlike ABC, there is no need to choose
a bandwidth $\epsilon$: this approximation is replaced with that
arising from the discrepancy between the normal approximation and
the exact distribution of the chosen summary statistic. The SL method remains
approximate even if the summary statistic distribution is Gaussian as $\widehat{f}_{\mbox{SL}}$ is not an unbiased
estimate of the density and so the exact-approximation results do not apply. Rather,
this is a special case of noisy MCMC, and we do not expect the additional
bias introduced by estimating the parameters of $\widehat{f}_{\mbox{SL}}$
to have large effects on the results, even if the parameters are estimated
via an internal MCMC chain targeting $f(\cdot|\theta)$ as described
in section \ref{sub:Russian-Roulette}.

SL is related to a number of other simulation based algorithms under
the umbrella of Bayesian indirect inference \citep{Drovandi2013}.
This suggests a number of extensions to some of the methods presented
in this paper that we do not explore here.

\subsection{Bayesian model comparison}
The main focus of this paper is estimating the \emph{marginal likelihood}
(also termed the \emph{evidence}) 
\[
p(y)=\int p(\theta)f(y|\theta)d\theta
\]
and \emph{Bayes' factors}:\emph{ }ratios of evidences for different
models ($M_{1}$ and $M_{2}$, say), $\mbox{BF}_{12}={p(y|M_{1})}/{p(y|M_{2})}$.
These quantities cannot usually be estimated reliably from MCMC output,
and commonly used methods for estimating them require $f(y|\theta)$
to be tractable in $\theta$. This leads \citet{Friel2013e} to label
their estimation as \emph{``triply intractable''} when $f$ has
an INC. To our knowledge the only published approach to estimating
the evidence for such models is in \citet{Friel2013e}, with this
paper also giving one of the only approaches to estimating BFs in
this setting. For estimating BFs, ABC provides a viable alternative
\citep{Grelaud2009}, at least for models within the exponential family.

\citet{Friel2013e} starts from Chib's approximation, 
\begin{equation}
\widehat{p}(y)=\frac{f(y|\widetilde{\theta})p(\widetilde{\theta})}{\widehat{\pi}(\widetilde{\theta}|y)},\label{eq:chib}
\end{equation}
where $\widetilde{\theta}$ can be an arbitrary value of $\theta$
and $\widehat{\pi}$ is an approximation to the posterior distribution.
Such an approximation is intractable when $f$ has an INC. \citet{Friel2013e}
devises a ``population'' version of the exchange algorithm that
simulates points $\theta^{(p)}$ from the posterior distribution,
and which also gives an estimate $\widehat{Z}(\theta^{(p)})$ of the
INC at each of these points. The points $\theta^{(p)}$ can be used
to find a kernel density approximation $\widehat{\pi}$, and estimates
$\widehat{Z}(\theta^{(p)})$ of the INC. These are then used in a
number of evaluations of \eqref{eq:chib} at points (generated
by the population exchange algorithm) in a region of high posterior
density, which are then averaged to find an estimate of the evidence.
This method has a number of useful properties (including that it may
be a more efficient approach for parameter inference than the standard
exchange algorithm), but for evidence estimation it suffers the limitation
of using a kernel density estimate which means that, as noted in the
paper, its use is limited to low-dimensional parameter spaces.

In this paper we explore the alternative approach of methods based
on IS, making use of the likelihood approximations described earlier
in this section. These IS methods are outlined in section \ref{sec:Importance-sampling-approaches}.
In section \ref{sec:Importance-sampling-approaches} we note the good empirical
performance of an inexact-approximate method and examine such approaches  in
more detail. As IS is itself not readily applicable to high dimensional parameter
spaces, in section \ref{sec:Sequential-Monte-Carlo} we look at natural
extensions to the IS methods based on SMC. Particular care is required when
considering approximations within iterative algorithms: we  provide a preliminary study
of approximation in this context demonstrating theoretically that the
resulting error can be controlled uniformly in time, under very favorable
assumptions. This, and the associated empirical study are intended to provide
motivation and proof of concept; caution is still required if approximation is
used within such methods in practice but the results presented suggest that
further investigation is warranted.  The algorithms presented later in the
paper are viable alternatives to the MCMC approaches to parameter estimation
described in this section, and may outperform the corresponding MCMC approach
in some cases. In particular they all automatically make use of a population
of points, an idea previously explored in the MCMC context by
\citet{Caimo2011} and \citet{Friel2013e}. In section \ref{sec:Conclusions} we
draw conclusions. 

\section{Importance sampling approaches\label{sec:Importance-sampling-approaches}}
In this section we investigate the use of IS for estimating the evidence
and BFs for models with INCs. We consider an ``ideal'' importance
sampler that simulates $P$ points $\left\{ \theta^{(p)}\right\} _{p=1}^{P}$
from a proposal $q(\cdot)$ and calculates their weight, in the presence
of an INC, using 
\begin{eqnarray}
\widetilde{w}^{(p)} & = & \frac{p(\theta^{(p)})\gamma(y|\theta^{(p)})}{q(\theta^{(p)})Z(\theta^{(p)})},\label{eq:is_z}
\end{eqnarray}
with an estimate of the evidence given by 
\begin{equation}
\widehat{p}(y)=\frac{1}{P}\sum_{p=1}^{P}\widetilde{w}^{(p)}.\label{eq:is_ml}
\end{equation}
To estimate a BF we simply take the ratio of estimates of the evidence
for the two models under consideration. However, the presence of the
INC in the weight expression in \eqref{eq:is_z} means that
importance samplers cannot be directly implemented for these models.
To circumvent this problem we will investigate the use of the techniques
described in section \ref{sub:Parameter-inference} in importance
sampling. We begin by looking at exact-approximation based methods
in section \ref{sub:Auxiliary-variable-IS}. We then examine the use
to approximate likelihoods based on simulation, including ABC and
SL in section \ref{sub:Simulation-based-methods}, before looking
at the performance of all of these methods on a toy example in section
\ref{sub:Toy-example}. Finally, in sections \ref{sub:Application-to-social}
and \ref{sub:Application-to-Ising} we examine applications to exponential
random graph models (ERGMs) and Ising models, the latter of which
leads us to consider the use of inexact-approximations in IS (first introduced in section \ref{sub:IS-with-biased}).

\subsection{Auxiliary variable IS\label{sub:Auxiliary-variable-IS}}
To avoid the evaluation of the INC in \eqref{eq:is_z}, we
propose the use of the auxiliary variable method used in the MCMC
context in section \ref{sub:Single-and-multiple}. Specifically,
consider IS using the SAV target 
\[
p(\theta,u|y)\propto q_{u}(u|\theta,y)f(y|\theta)p(\theta),
\]
noting that it has the same normalizing constant as
$p(\theta|y) \propto f(y|\theta)p(\theta)$, with proposal 
\[
q(\theta,u)=f(u|\theta)q(\theta).
\]
This results in weights 
\begin{eqnarray*}
\widetilde{w}^{(p)} & = & \frac{q_{u}(u|\theta^{(p)},y)\gamma(y|\theta^{(p)})p(\theta^{(p)})}{\gamma(u|\theta^{(p)})q(\theta^{(p)})}\frac{Z(\theta^{(p)})}{Z(\theta^{(p)})}\\
 & = & \frac{\gamma(y|\theta^{(p)})p(\theta^{(p)})}{q(\theta^{(p)})}\frac{q_{u}(u|\theta^{(p)},y)}{\gamma(u|\theta^{(p)})},
\end{eqnarray*}
and the estimate \eqref{eq:is_ml}
of the evidence.

In this method, which we will refer to as single auxiliary variable
IS (SAVIS), we may view \linebreak $q_{u}(u|\theta^{(p)},y)/\gamma(u|\theta^{(p)})$
as an unbiased importance sampling (IS) estimator of $1/Z(\theta^{(p)})$.
Although we are using an unbiased estimator of the weights in place
of the ideal weights, the result is still an exact importance sampler.
SAVIS is an exact-approximate IS method, as seen previously in \citet{Fearnhead2010f},
\citet{Chopin2013} and \citet{Tran2013}. As in the MCMC setting,
to ensure the variance of estimators produced by this scheme is not
large we must ensure the variance of estimator of $1/Z(\theta^{(p)})$
is small. Thus in practice we found extensions to this basic algorithm
were useful: using multiple $u$ importance points for each proposed
$\theta^{(p)}$ as in \eqref{eq:1overz_sav_multiplepoints};
and using AIS, rather than simple IS, for estimating $1/Z(\theta^{(p)})$
as in \eqref{eq:1overz_mav_multiplepoints} (giving an algorithm
that we refer to as multiple auxiliary variable IS (MAVIS), in common
with the terminology in \citet{Murray2006}). Using $q_{u}(\cdot|\theta,y)=f(\cdot|\widehat{\theta})$,
as described in section \ref{sub:Single-and-multiple}, and $\gamma_{k}$
as in \eqref{eq:ais_seq}, we obtain 
\begin{equation}
\widehat{\frac{1}{Z(\theta)}}=\frac{1}{Z(\widehat{\theta})}\frac{1}{M}\sum_{m=1}^{M}\prod_{k=1}^{K}\frac{\gamma_{k}(u_{k-1}^{(m)}|\theta^{*},\theta,y)}{\gamma_{k-1}(u_{k-1}^{(m)}|\theta^{*},\theta,y)}.\label{eq:1overz_mavis}
\end{equation}
In this case the (A)IS methods are being used as unbiased estimators
of the ratio $Z(\widehat{\theta})/Z(\theta)$ and again SMC could be used in their place.

\selectlanguage{british}%

\subsection{Simulation based methods\label{sub:Simulation-based-methods}}

\selectlanguage{english}%
\citet{Didelot2011} investigate the use of the ABC approximation
when using IS for estimating marginal likelihoods. In this case the
weight equation becomes 
\[
\widetilde{w}^{(p)}=\frac{p(\theta^{(p)})\frac{1}{R}\sum_{r=1}^{R}\pi_{\epsilon}(S(y)|S(x_{r}^{(p)}))}{q(\theta^{(p)})},
\]
where $\left\{ x_{r}^{(p)}\right\} _{r=1}^{R}\sim f(\cdot|\theta^{(p)})$,
and using the notation from section \ref{sub:Approximate-Bayesian-computation}.
However, using these weights within \eqref{eq:is_ml} gives
an estimate for $p(S(y))$ rather than, as desired, an estimate of
the evidence $p(y)$. 

Fortunately, there are cases in which ABC may be used to estimate
BFs. \citet{Didelot2011} establish that, for the BF for two exponential
family models: if $S_{1}(y)$ is sufficient for the parameters in
model 1 and $S_{2}(y)$ is sufficient for the parameters in model
2, then using $S(y)=(S_{1}(y),S_{2}(y))$ gives 
\[
\frac{p(y|M_{1})}{p(y|M_{2})}=\frac{p(S(y)|M_{1})}{p(S(y)|M_{2})}.
\]
Outside the exponential family, making an appropriate choice of summary
statistics is harder \citep{Robert2011j,Prangle2013,Marin2013}.

Just as in the parameter estimation case, the use of a tolerance $\epsilon>0$
results in estimating an approximation to the true BF. An alternative
approximation, not previously used in model comparison, is to use
SL (as described in section \ref{sub:Synthetic-likelihood}). In this
case the weight equation becomes 
\[
\widetilde{w}^{(p)}=\frac{p(\theta^{(p)})\mathcal{N}\left(S(y);\widehat{\mu}_{\theta^{(p)}},\widehat{\Sigma}_{\theta^{(p)}}\right)}{q(\theta^{(p)})},
\]
where $\widehat{\mu}_{\theta},\widehat{\Sigma}_{\theta}$ are given
by \eqref{eq:sl_params}. As in parameter
estimation, this approximation is only appropriate if the normality
assumption is reasonable. The choice of summary statistics is as difficult as
in the ABC case. 

\subsection{Toy example\label{sub:Toy-example}}

In this section we have discussed three alternative methods for estimating
BFs: MAVIS, ABC and SL. To further understand their properties we
now investigate the performance of each method on a toy example.

Consider i.i.d. observations $y=\left\{ y_{i}\right\} _{i=1}^{n=100}$
of a discrete random variable that takes values in $\mathbb{N}$.
For such a dataset, we will find the BF for the models 
\begin{enumerate}
\item $y|\theta\sim\mbox{Poisson}(\theta)$, $\theta = \lambda \sim\mbox{Exp}(1)$
\begin{eqnarray*}
f_{1}\left(y|\theta\right) 
& = & \prod_{i=1}^{n}\frac{\lambda^{y_{i}}}{y_{i}!} / \exp(-n\lambda)
\end{eqnarray*}

\item $y|\theta\sim\mbox{Geometric}(\theta)$, $\theta = p \sim\mbox{Unif}(0,1)$
\begin{eqnarray*}
f_{2}\left(y|\theta\right) 
 & = & \prod_{i=1}^{n}(1-p)^{y_{i}} / p^{-n}.
\end{eqnarray*}

\end{enumerate}
In both cases we have rewritten the likelihoods $f_{1}$ and $f_{2}$
in the form $\gamma(y|\theta)/Z(\theta)$ in order to use MAVIS. Due
to the use of conjugate priors the BF for these two models can be
found analytically. As in \citet{Didelot2011} we simulated (using
an approximate rejection sampling scheme) 1000 datasets for which
$\frac{p(y|M_{1})}{p(y|M_{1})+p(y|M_{2})}$ roughly uniformly cover
the interval {[}0.01,0.99{]}, to ensure that testing is performed
in a wide range of scenarios. For each algorithm we used the same
computational effort, in terms of the number of simulations ($100,000$)
from the likelihood (such simulations dominate the computational cost of all
of the algorithms considered). 

Our results are shown in figure \ref{fig:Bayes'-factors-for}, with
the algorithm-specific parameters being given in figure \ref{fig:A-box-plot}.
We note that we achieved better results for MAVIS when: devoting more
computational effort to the estimation of $1/Z(\theta)$ (thus we
used only 100 importance points in $\theta$-space, compared to 1000
for the other algorithms); and using more intermediate bridging distributions
in the AIS, rather than multiple importance points (thus, in equation
\eqref{eq:1overz_mavis} we used $K=1000$ and $M=1$). In the ABC case
we found that reducing $\epsilon$ much further than 0.1 resulted
in many importance points with zero weight (note that here, and throughout the paper we use the uniform kernel for $\pi_{\epsilon}$). From the box plots in
figure \ref{fig:A-box-plot}, we might infer that overall SL has outperformed
the other methods, but be concerned about the number of outliers.
Figures \ref{fig:The--of} to \ref{fig:The--of-2} shed more light
on the situations in which each algorithm performs well.

\begin{figure*}
\center
\subfloat[A box plot of the $\log$ of the estimated BF divided by the true
BF.\label{fig:A-box-plot}]{

\includegraphics[scale=0.27]{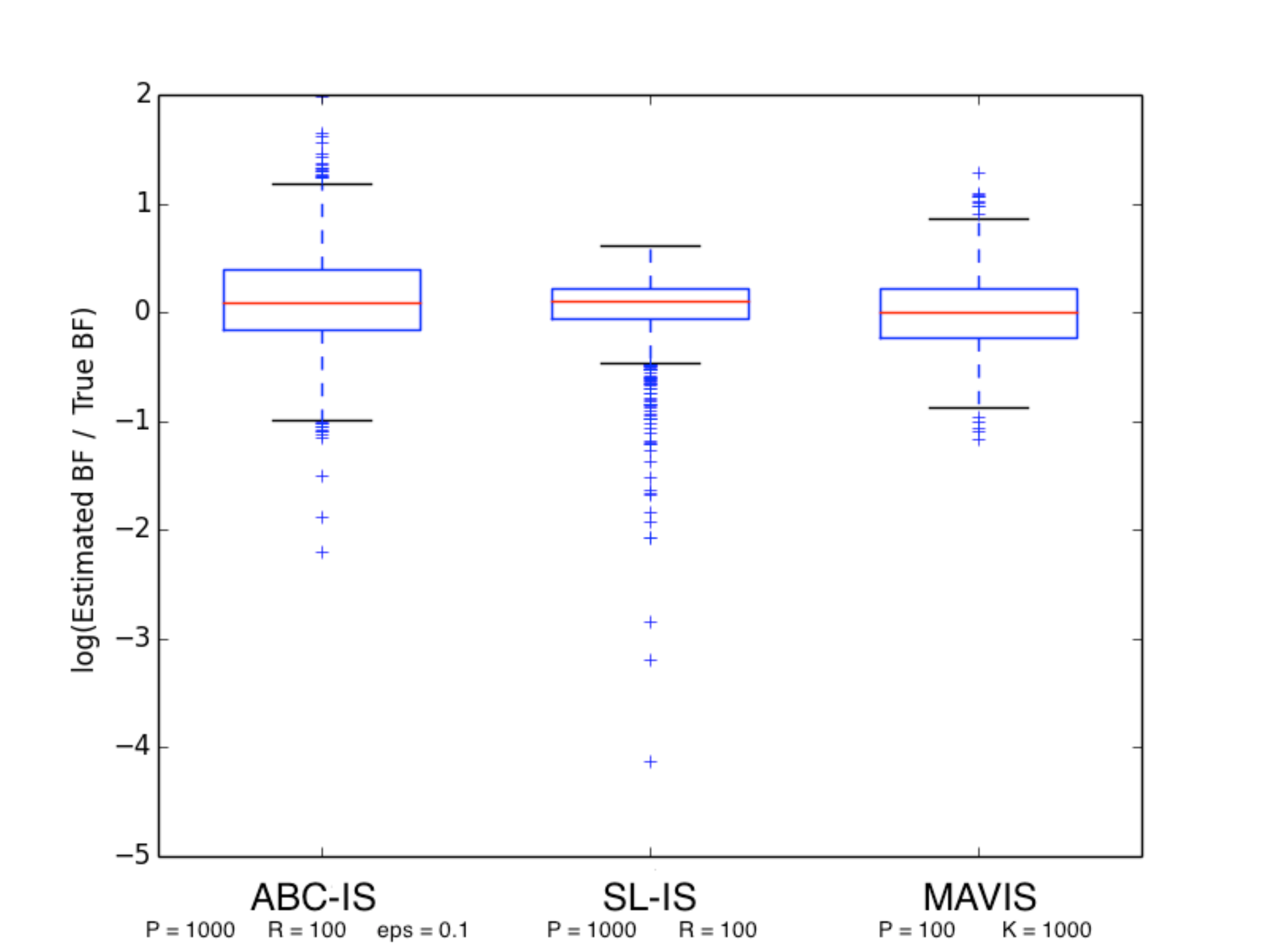}}\hspace*{5mm}\subfloat[The $\log$ of the BF estimated by ABC-IS against the $\log$ of the
true BF.\label{fig:The--of}]{

\selectlanguage{english}%
\includegraphics[scale=0.38]{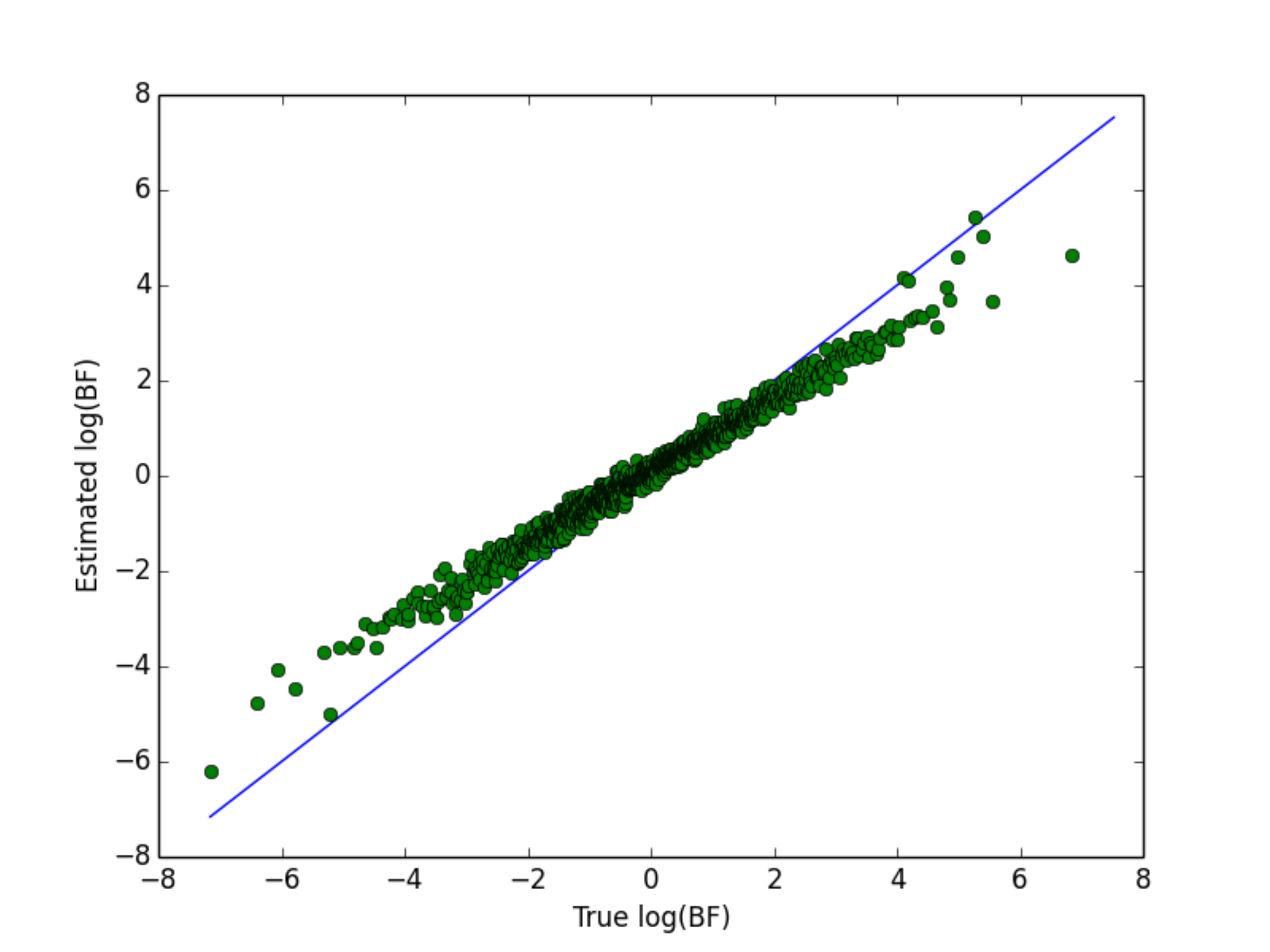}\selectlanguage{british}%
}

\subfloat[The $\log$ of the BF estimated by SL-IS against the $\log$ of the
true BF.\label{fig:The--of-1}]{

\selectlanguage{english}%
\includegraphics[scale=0.38]{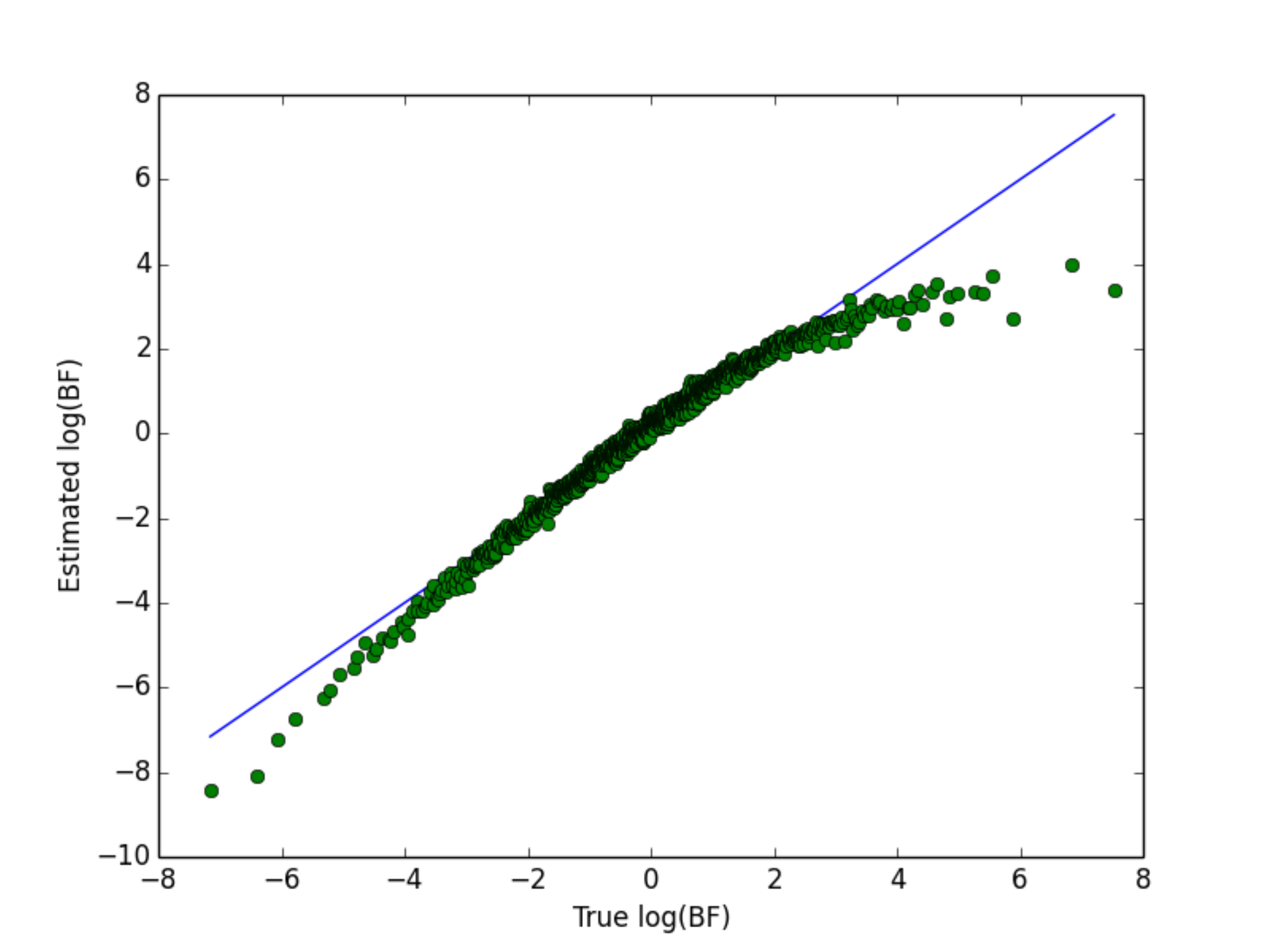}\selectlanguage{british}%
}\hspace*{5mm}\subfloat[The $\log$ of the BF estimated by MAVIS against the $\log$ of the
true BF.\label{fig:The--of-2}]{

\selectlanguage{english}%
\includegraphics[scale=0.38]{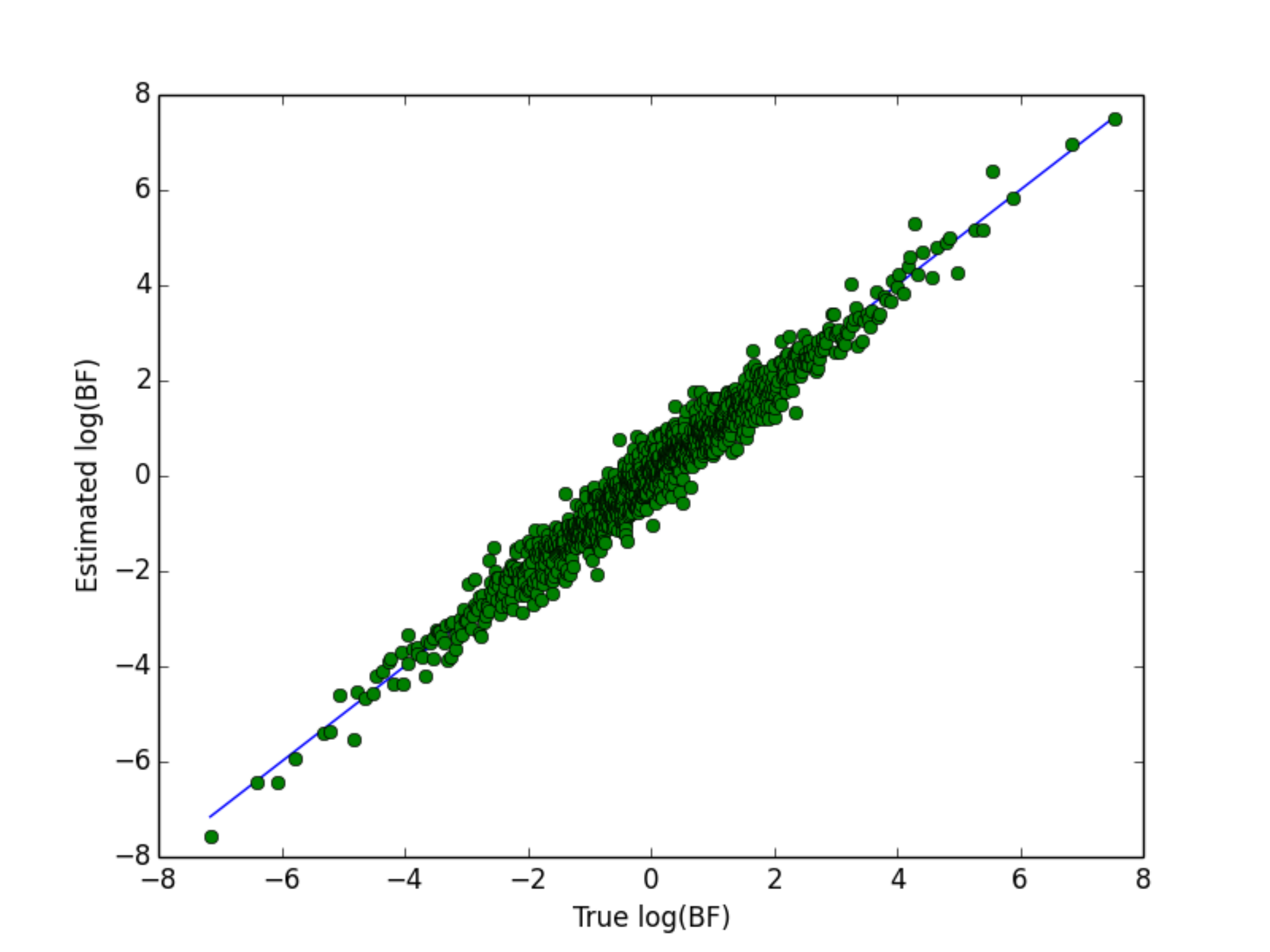}\selectlanguage{british}%
}\protect\caption{Bayes' factors for the Poisson and geometric models.\label{fig:Bayes'-factors-for}}
\end{figure*}

In figure \ref{fig:The--of} we observe that the non-zero $\epsilon$
results in a bias in the BF estimates (represented by the shallower
slope in the estimated BFs compared to the true values). In this example
we conclude that ABC has worked quite well, since the bias is only
pronounced in situations where the true BF favours one model strongly
over the other, and this conclusion would not be affected by the bias.
For this reason it might be more relevant in this example to consider
the deviations from the shallow slope, which are likely due to the
Monte Carlo variance in the estimator (which becomes more pronounced
as $\epsilon$ is reduced). We see that the choice of $\epsilon$
essentially governs a bias-variance trade-off, and that the difficulty
in using the approach more generally is that it is not easy to evaluate
whether a choice of $\epsilon$ that ensures a low variance also ensures
that the bias is not significant in terms of affecting the conclusions
that might be drawn from the estimated BF (see section \ref{sub:Application-to-social}).
Figure \ref{fig:The--of-1} suggests that SL has worked extremely
well (in terms of having a low variance) for the most important situations,
where the BF is close to 1. However, we note that the large biases
introduced due to the limitation of the Gaussian assumption when the
BF is far from 1. Figure \ref{fig:The--of-2} indicates that there
is little or no bias when using MAVIS, but that there is appreciable variance (due
to using IS on the relatively high-dimensional $u$-space).

These results highlight that the three methods will be most effective
in slightly different situations. The approximations in ABC and SL
introduce a bias, the effect of which might be difficult to assess.
In ABC (assuming sufficient statistics) this bias can be reduced by
an increased computational effort allowing a smaller $\epsilon$,
however it is essentially impossible to assess when this bias is ``small
enough''. SL is the simplest method to implement, and seems to work
well in a wide variety of situations, but the advice in \citet{Wood2010f}
should be followed in checking that the assumption of normality is
appropriate. MAVIS is limited by the need to perform importance sampling
on the high-dimensional $(\theta,u)$ space 
but consequently avoids specifying summary statistics, its bias is small, and this
method is able to estimate the evidence of individual models.

\subsection{Application to social networks\label{sub:Application-to-social}}

In this section we use our methods to compare the evidence for two
alternative ERGMs for the Gamaneg data previously analysed in
\citet{Friel2013e} (who illustrate the data in their figure 3). An ERGM has
the general form
\[
f(y|\theta)=\frac{1}{Z(\theta)}\exp\left(\theta^{T}S(y)\right),
\]
where $S(y)$ is a vector of statistics of a network $y$ and $\theta$
is a parameter vector of the same length. We take $S(y)=$ $(\#$ of edges $)$
in model 1 and $S(y)=($\# of edges, \# of two stars$)$ in
model 2 . As in \citet{Friel2013e} we use the prior $p(\theta)=\mathcal{N}(\theta;0,25I)$.

%

\selectlanguage{english}%
Using a computational budget of $10^{5}$ simulations from the likelihood
(each simulation consisting of an internal MCMC run of length 1000
as a proxy for an exact sampler, as described in section \ref{sub:Russian-Roulette}),
\citet{Friel2013e} finds that the evidence for model 1 is $\sim37\times$
that for model 2. Using the same computational budget for our methods,
consisting of 1000 importance points (with 100 simulations from the
likelihood for each point), we obtained the results shown in Table
\ref{tab:gamaneg}. 

\selectlanguage{british}%

\begin{table}
\begin{centering}
\begin{tabular}{|c|c|c|c|c|}
\hline 
 & ABC ($\epsilon=0.1$) & ABC ($\epsilon=0.05$) & SL & MAVIS\tabularnewline
\hline 
\hline 
$\frac{\hat{p}(y|M_{1})}{\hat{p}(y|M_{2})}$ & 4 & 20 & 40 & 41\tabularnewline
\hline 
\end{tabular}
\par\end{centering}

\protect\caption{Model comparison results for Gamaneg data. Note that the ABC ($\epsilon=0.05$)
estimate was based upon just 5 sample points of non-zero weight. MAVIS
also provides estimates of the individual evidence ($\log\left[\widehat{p}(y|M_{1})\right]=-69.6$,
$\log\left[\widehat{p}(y|M_{2})\right]=-73.3$).\label{tab:gamaneg}}
\end{table}

This example highlights the issue with the bias-variance trade-off
in ABC, with $\epsilon=0.1$ having too large a bias and $\epsilon=0.05$
having too large a variance. SL performs well --- in this particular
case the Gaussian assumption appears to be appropriate. One might
expect this, since the statistics are sums of random variables. However,
we note that this is not usually the case for ERGMs, particularly
when modelling large networks, and that SL is a much more appropriate
method for inference in the ERGMs with local dependence \citep{Schweinberger2015}.
A more sophisticated ABC approach might exhibit improved
performance, possibly outperforming SL. However, the appeal of SL
is in its simplicity, and we find it to be a useful method for obtaining
good results with minimal tuning.

\subsection{IS with biased weights\label{sub:IS-with-biased}}
The implementation of MAVIS in the previous section is not an exact-approximate
method for two reasons: 
\begin{enumerate}
\item An internal MCMC chain was used in place of an exact sampler; 
\item The $1/Z(\widehat{\theta})$ term in \eqref{eq:1overz_mavis}
was estimated before running this algorithm (by using a standard SMC
method, with initial distribution being the Bernoulli random graph (which can be simulated from exactly) and final distribution $\propto \gamma(\cdot | \widehat{\theta})$ to estimate $Z(\widehat{\theta})$ (being the normalising constant of $\gamma$), and taking the reciprocal)
with this fixed estimate being used throughout. 
\end{enumerate}

However, in practice, we tend to find that such ``inexact-approximations''
do not introduce large errors into \linebreak Bayes' factor estimates, particularly
when compared to standard implementations of ABC (as seen in the previous
section). 

This example suggests that 
in practice it may sometimes be advantageous to use biased rather
than unbiased estimates of importance weights within a random weight
IS algorithm: an observation that is somewhat analogous to that made
in \citet{Alquier2014} in the context of MCMC. This section provides
an initial theoretical exploration as to whether this might be a useful strategy
in IS. 

In order to analyse the behaviour of importance sampling with biased
weights, we consider biased estimates of the weights in equation \eqref{eq:is_ml}. Let
\[
w(\theta):=\frac{p(\theta)\gamma(y|\theta)}{Z(\theta)q(\theta)}.
\]
We consider biased randomised weights that admit an additive decomposition,
\[
\grave{w}(\theta):= w(\theta) + b(\theta)+\grave{V}_{\theta},
\]
in which $b(\theta) = \mathbb{E}[\grave{w}(\theta)|\theta] - w(\theta)$ is a deterministic function describing the
bias of the weights and $\grave{V}_{\theta}$ is a random variable (more
precisely, there is an independent copy of such a random variable
associated with every particle), which conditional upon $\theta$
is of mean zero and variance $\grave{\sigma}_{\theta}^{2} = \textrm{Var}(\grave{w}(\theta)|\theta)$. This decomposition
will not generally be available in practice, but is flexible enough
to allow the formal description of many settings of interest. For
instance, one might consider the algorithms presented here 
by setting $b(\theta)$ to the (conditional) expected value of the
difference between the approximate and exact weights and $\grave{V}_{\theta}$
to the difference between the approximate weights and their expected
value.

We have immediately that the bias of such an estimate is, using a subscript of
$q$ to denote expectations and variances with respect to $q(\theta)$, $\mathbb{E}_{q}[b(\theta)]$.
By a simple application of the law of total variance, its variance is 
\[
\frac{1}{P}\textrm{Var}_{q}(\grave{w}(\theta)) = \frac{1}{P}\left\{ \textrm{Var}_{q}\left[w(\theta)+b(\theta)\right]+\mathbb{E}_{q}\left[\grave{\sigma}_{\theta}^{2}\right]\right\} 
\]
Consequently, the mean squared error of this estimate is: 
\begin{eqnarray*}
\frac{1}{P}\left\{ \textrm{Var}_{q}\left[w(\theta)+b(\theta)\right]+\mathbb{E}_{q}[\grave{\sigma}_{\theta}^{2}]\right\} +\mathbb{E}_{q}[b(\theta)]^{2}.\nonumber
\end{eqnarray*}
If we compare such a biased estimator with a second estimator in which
we use the same proposal distribution but instead use an unbiased random
weight 
\[
\acute{w}(\theta):=w(\theta)+\acute{V}(\theta),
\]
where $\acute{V}(\theta)$ has conditional expectation zero and
variance $\acute{\sigma}_{\theta}^{2}$, then it's clear that the
biased estimator has smaller mean squared error for small enough samples
if it has sufficiently smaller variance, i.e., when (assuming \foreignlanguage{english}{$\mathbb{E}_{q}[b(\theta)]^{2}>0$},
otherwise one estimator dominates the other for all sample sizes):
%
%
\begin{align*}
& \frac{1}{P}\left\{ \textrm{Var}_{q}\left[w(\theta)+b(\theta)\right]+\mathbb{E}_{q}[\grave{\sigma}_{\theta}^{2}]\right\} +\mathbb{E}_{q}[b(\theta)]^{2}\nonumber\\
<&
\frac{1}{P}\left\{ \textrm{Var}_{q}\left[w(\theta)\right]+\mathbb{E}_{q}[\acute{\sigma}_{\theta}^{2}]\right\} \nonumber
\end{align*}
which holds when $P$ is inferior to
{
\begin{eqnarray}
\frac{\mathbb{E}_{q}[\acute{\sigma}_{\theta}^{2}-\grave{\sigma}_{\theta}^{2}]-\textrm{Var}_{q}\left[b(\theta)\right]-2\textrm{Cov}_{q}\left[w(\theta),b(\theta)\right]}{\mathbb{E}_{q}[b(\theta)]^{2}}.\nonumber
\end{eqnarray}}

In the artificially simple setting in which $b(\theta)=b_{0}$ is
constant, this would mean that the biased estimator would have smaller
MSE for samples smaller than the ratio of the difference in variance
to the square of that bias suggesting that qualitatively a biased
estimator might be better if the square of the average bias is small
in comparison to the variance reduction that it provides. Given a
family of increasingly expensive biased estimators with progressively
smaller bias, one could envisage using such an argument to manage
the trade-off between less biased estimators and larger sample sizes. In
practice a negative covariance between $b(\theta)$ and $w(\theta)$ might also lead to
favourable performance by biased estimators.


\subsection{Applications to Ising models\label{sub:Application-to-Ising}}
In the current section we investigate this type of approach further empirically, estimating Bayes' factors from data simulated from
Ising models. In particular we reanalyse the data from \citet{Friel2013e},
which consists of 20 realisations from a first-order $10\times10$
Ising model and 20 realisations from a second-order $10\times10$
Ising model for which accurate estimates (via \citet{Friel2007})
of the evidence serve as a ground truth for comparison. We also analyse data from a $100\times100$ Ising model. 

\subsubsection{$10\times 10$ Ising models\label{sub:10by10}}

As in the toy example, we examine several different configurations of the IS and AIS estimators of the $Z(\widehat{\theta})/Z(\theta)$ term in the weight \eqref{eq:is_z}, using different values of $M$, $K$ and $B$, the burn in of the internal MCMC, that yield the same computational cost (in terms of the number of Gibbs sweeps used to simulate from the likelihood). Note that for small values of $B$ these estimators are biased; a bias that decreases as $B$ increases.

The empirical results in \citet{Friel2013e}, use a total $2 \times 10^7$ Gibbs sweeps to estimate one Bayes' factor, to allow comparison of
our results with those in that paper. Here, estimating 
a marginal likelihood is done in three stages: firstly $\widehat{\theta}$
is estimated; followed by $Z(\widehat{\theta})$, then finally the
marginal likelihood. We took $\widehat{\theta}$ to be the posterior
expectation, estimated from a run of the exchange algorithm of $10,000$
iterations. $Z(\widehat{\theta})$ was then estimated using SMC with
an MCMC move, with 200 particles and 100 targets, with the $i$th
target being $\gamma_{i}(\cdot|\theta)=\gamma_{i}\left(\cdot|{i\theta}/{100}\right)$,
employing stratified resampling when the effective sample size (ESS;
\citet{Kong1994}) falls below 100. The total cost of these three stages is $5 \times 10^6$ Gibbs sweeps ($1/4$ of the cost of population exchange) with the final IS stage costing $2 \times 10^4$ sweeps ($1/1000$ of the cost of population exchange). We note that the cost of the first two stages has been chosen conservatively - less computational effort here can also yield good results. The importance proposal used in
all cases was a multivariate normal distribution, with mean and variance
taken to be the sample mean and variance from the initial run of the
exchange algorithm. This proposal would clearly not be appropriate
in high dimensions, but is reasonable for the low dimensional parameter
spaces considered here.  Figure \ref{fig:Box-plots-of} shows the results produced
by these methods in comparison with those from \citet{Friel2013e}.

\begin{figure*}
\includegraphics[scale=0.55]{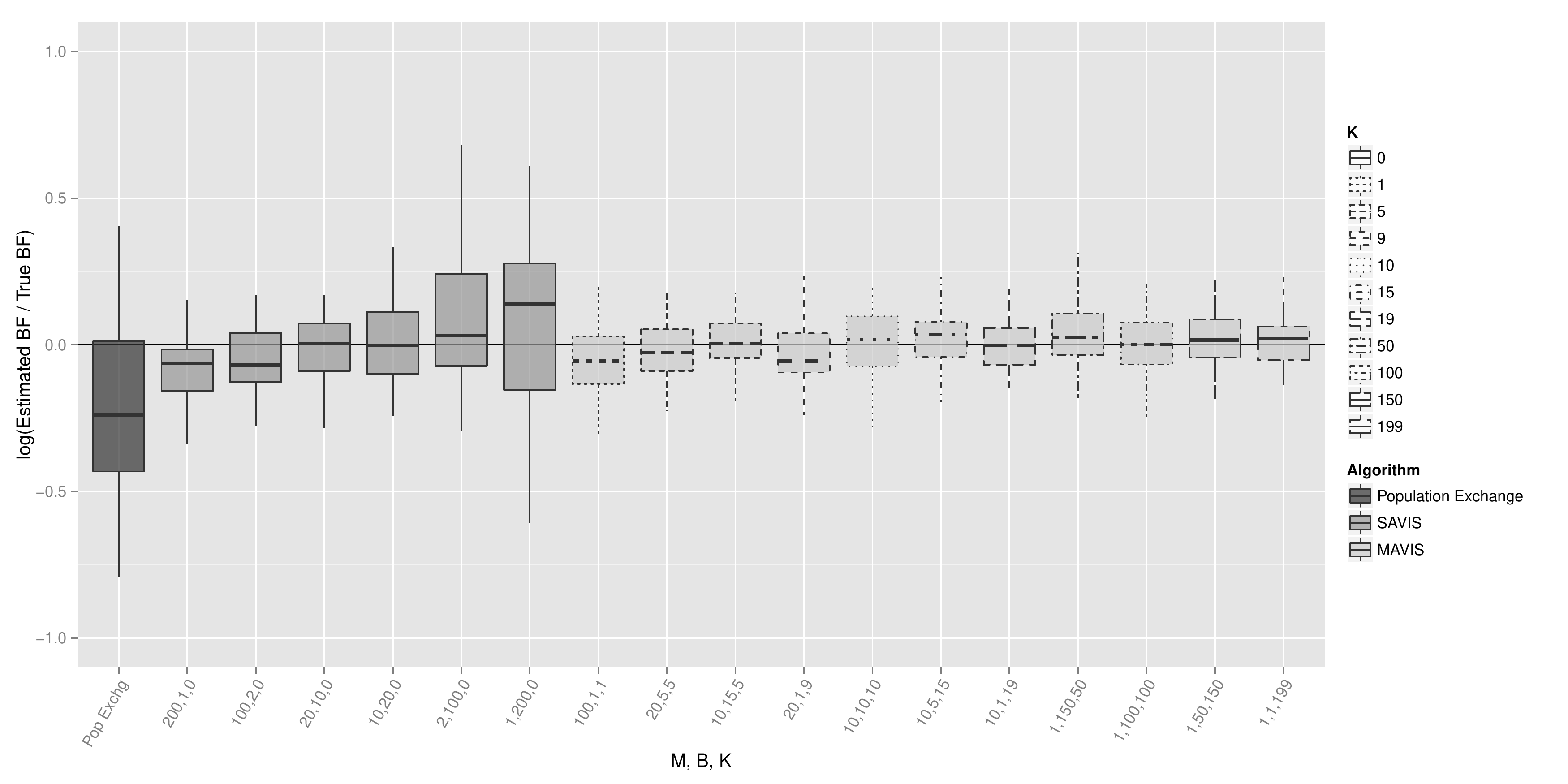}

\protect\caption{Box plots of the results of population exchange, SAVIS, and MAVIS on the Ising data.\label{fig:Box-plots-of}}
\end{figure*}

We observe: improvements of the new methods over population exchange; an overall robustness of the new methods to different choices of parameters; and that there is a bias-variance tradeoff in the ``internal'' estimate of $Z(\widehat{\theta})/Z(\theta)$ in terms of producing the best behaviour of the Bayes' factor estimates. Recall that as $B$ increases the bias of the internal estimate (the results of which can be observed in the results when using $B=0$) decreases, but for a fixed computational effort it is beneficial to use a lower $B$ and to instead increase $M$, using more importance points to decrease the variance. As in \citet{Alquier2014}, we observe that it may be useful to move away from the exact-approximate approaches, and in this case, to simply
use the best available estimator of $Z(\widehat{\theta})/Z(\theta)$ (taking
into account its statistical and computational efficiency) regardless
of whether it is unbiased. In this example there is little observed difference in using our fixed computational budget on more AIS moves ($K$) in place of using more importance points ($M$). In general we might expect using more AIS moves to be more productive when the estimates of the $Z(\widehat{\theta})/Z(\theta)$ for $\theta$ far from $\widehat{\theta}$ are required.

\subsubsection{$100\times 100$ Ising model\label{sub:100by100}}

In this section we use SAVIS for estimating the marginal likelihood for a first order Ising model on data of size $100\times 100$ pixels simulated from an Ising model with parameter $\theta=10$. Again, estimating a marginal likelihood is done in three stages: firstly $\widehat{\theta}$
is estimated; followed by $Z(\widehat{\theta})$, then finally the
marginal likelihood. The methods use for the first two stages are identical to those used in section \ref{sub:10by10}, as is the choice of proposal distribution. The third stage is performed using SAVIS with $M=100$ and $B=20$. From 20 runs of this third stage, a five-number summary of the $\log$ evidence estimates was (-5790.251, -5790.178, -5790.144, -5790.119, -5790.009), with the average ESS being 80.75. Note the low variance over these runs of the algorithm and the high ESS, which were also found for different configurations of the algorithm (including for more importance points and larger values of $M$ and $B$). One might expect this example to be more difficult than the $10\times 10$ grids considered in the previous section, due to the need to find good estimates of $Z(\widehat{\theta})/Z(\theta)$ that are here normalising constants of distributions on a space of higher dimensions. However, since the posterior has lower variance in this case, only values of $\theta$ close to $\widehat{\theta}$ are proposed, which makes estimating $Z(\widehat{\theta})/Z(\theta)$ much easier, yielding the good results in this section.

\subsection{Discussion}

In this section we have compared the use of ABC-IS, SL-IS, MAVIS (and
alternatives) for estimating marginal likelihoods and Bayes' factors.
The use of ABC for model comparison has received much attention, with
much of the discussion centring around appropriate \linebreak choices of summary
statistics. We have avoided this in our examples by using exponential
family models, but in general this remains an issue affecting both
ABC and SL. It is the use of summary statistics that makes ABC and
SL unable to provide evidence estimates. However, it is the use of
summary statistics, usually essential in these settings, that provides
ABC and SL with an advantage over MAVIS, in which importance sampling
must be performed over the high dimensional data-space. Despite this
disadvantage, MAVIS avoids the approximations made in the simulation
based methods (illustrated in figures \ref{fig:The--of} to \ref{fig:The--of-2},
with the accuracy depending primarily on the quality of the estimate
of $1/Z$ used). In section \ref{sub:Application-to-Ising} we saw
that there can be advantages of using biased, but lower variance estimates
in place of standard IS.

The main weakness of all of the methods described in this section
is that they are all based on standard IS and are thus not practical
for use when $\theta$ is high dimensional. In the next section we
examine the use of SMC samplers as an extension to IS for use on triply
intractable problems, and in this framework discuss further the effect
of inexact approximations.

\section{Sequential Monte Carlo approaches\label{sec:Sequential-Monte-Carlo}}

SMC samplers \citep{DelMoral2006c} are a generalisation of IS, in
which the problem of choosing an appropriate proposal distribution
in IS is avoided by performing IS sequentially on a sequence of target
distributions, starting at a target that is easy to simulate from,
and ending at the target of interest. In standard IS the number of
Monte Carlo points required in order to obtain a particular accuracy
increases exponentially with the dimension of the space, but \citet{Beskos2011c}
show (under appropriate regularity conditions) that the use of SMC
circumvents this problem and can thus be practically useful in high
dimensions.

In this section we introduce SMC algorithms for simulating from doubly
intractable posteriors which have the by-product that, like IS, they
also produce estimates of marginal likelihoods. \foreignlanguage{english}{We
note that, although here we focus on estimating the evidence, the
SMC sampler approaches based here are a natural alternative to the
MCMC methods described in section \ref{sub:Parameter-inference}.
and inherently use a ``population'' of Monte Carlo points (shown
to be beneficial on these models by \citet{Caimo2011}).} In section
\ref{sub:SMC-samplers-in} we describe these algorithms, before examining
an application to estimating the precision matrix of a Gaussian distribution
in high dimensions in section \ref{sub:Application-to-precision}. In
\ref{sec:biased-smc} we provide a preliminary investigation  of the
consequences of using biased weight estimates in an SMC framework.

\subsection{SMC samplers in the presence of an INC\label{sub:SMC-samplers-in}}

This section introduces two alternative SMC samplers for use on doubly
intractable target distributions. The first, marginal SMC, directly
follows from the IS methods in the previous section. The second, SMC-MCMC,
requires a slightly different approach, but is more computationally
efficient. Finally we briefly discuss \linebreak simulation-based SMC samplers
in section \ref{sub:Simulation-based-SMC-samplers}.

To begin, we introduce notation that is common to all algorithms that
we discuss. SMC samplers perform sequential IS using $P$ ``particles''
$\theta^{(p)}$, each having (normalised) weight $w^{(p)}$, using
a sequence of targets $\pi_{0}$ to $\pi_{T}$, with $\pi_{T}$ being
the distribution of interest, in our case $\pi(\theta|y)\propto p(\theta)f(y|\theta)$.
In this section we will take $\pi_{t}(\theta|y)\propto p(\theta)f_{t}(y|\theta)=p(\theta)\gamma_{t}(y|\theta)/Z_{t}(\theta)$.
At target $t$, a ``forward'' kernel $K_{t}(\cdot|\theta_{t-1}^{(p)})$
is used to move particle $\theta_{t-1}^{(p)}$ to $\theta_{t}^{(p)}$,
with each particle then being reweighted to give unnormalised weight
\begin{eqnarray*}
\widetilde{w}_{t}^{(p)} & = & 
\frac{p(\theta_{t}^{(p)})\gamma_{t}(y|\theta_{t}^{(p)})}{p(\theta_{t-1}^{(p)})\gamma_{t-1}(y|\theta_{t-1}^{(p)})}\frac{Z_{t-1}(\theta_{t-1}^{(p)})}{Z_{t}(\theta_{t}^{(p)})}
\frac{L_{t-1}(\theta_{t}^{(p)},\theta_{t-1}^{(p)})}{K_{t}(\theta_{t-1}^{(p)},\theta_{t}^{(p)})}.
\end{eqnarray*}
Here, $L_{t-1}$ represents a ``backward'' kernel that we chose
differently in the alternative algorithms below. We note the presence
of the INC, which means that this algorithm cannot be implemented
in practice in its current form. The weights are then normalised to
give $\left\{ w_{t}^{(p)}\right\} $, and a resampling step is carried
out. In the following sections the focus is on the reweighting step: this is the main difference between the different algorithms.
For more detail on these methods, see \citet{DelMoral2007a}. %

\citet{Zhou2013} describe how BFs can be estimated directly by SMC
samplers, simply by taking $\pi_{1}$ to be one model and $\pi_{T}$
to be the other (with the $\pi_{t}$ being intermediate distributions).
This idea is also explored for Gibbs random fields in \citet{Friel2013e}.
However, the empirical results in \citet{Zhou2013} suggest that in
some cases this method does not necessarily perform better than estimating
marginal likelihoods for the two models separately and taking the
ratio of the estimates. Here we do not investigate these algorithms
further, but note that they offer an alternative to estimating the
marginal likelihood separately.

\subsubsection{Random weight SMC Samplers\label{sub:Random-weight-SMC}}

\paragraph{SMC with an MCMC kernel\label{sub:SMC-with-an}}

Suppose we were able to use a reversible MCMC kernel $K_{t}$ with
invariant distribution $\pi_{t}(\theta|y)\propto p(\theta)f_{t}(y|\theta)$,
and choose the $L_{t-1}$ kernel to be the time reversal of $K_{t}$
with respect to its invariant distribution, we obtain the following
incremental weight: 
\begin{eqnarray}
\widetilde{w}_{t}^{(p)} 
 & = & \frac{\gamma_{t}(y|\theta_{t-1}^{(p)})}{\gamma_{t-1}(y|\theta_{t-1}^{(p)})}\frac{Z_{t-1}(\theta_{t-1}^{(p)})}{Z_{t}(\theta_{t-1}^{(p)})}.\label{eq:smc-mcmc}
\end{eqnarray}
Once again, we cannot evaluate this incremental weight due to the
presence of a ratio of normalising constants. Also, such an MCMC kernel
cannot generally be directly constructed --- the MH update itself involves
evaluating the ratio of intractable normalising constants. However,
appendix \ref{sec:Using-SAV-and} shows that precisely the same weight
update results when using either SAV or exchange MCMC moves in place
of a direct MCMC step.

In order that this approach may be implemented we might consider,
in the spirit of the approximations suggested in section \ref{sec:Importance-sampling-approaches},
using an estimate of the ratio term $Z_{t-1}(\theta_{t-1}^{(p)})/Z_{t}(\theta_{t-1}^{(p)})$.
For example, an unbiased IS estimate is given by 
\begin{equation}
\widehat{\frac{Z_{t-1}(\theta_{t-1}^{(p)})}{Z_{t}(\theta_{t-1}^{(p)})}}=\frac{1}{M}\sum_{m=1}^{M}\frac{\gamma_{t-1}(u_{t}^{(m,p)}|\theta_{t-1}^{(p)})}{\gamma_{t}(u_{t}^{(m,p)}|\theta_{t-1}^{(p)})},\label{eq:smc_unbiased_Z}
\end{equation}
where $u_{t}^{(m,p)}\sim f_{t}(\cdot|\theta_{t-1}^{(p)})$. Although
this estimate is unbiased, we note that the resultant algorithm does
not have precisely the same extended space interpretation as the methods
in \citet{DelMoral2006c}. Appendix \ref{sec:An-Extended-Space} gives
an explicit construction for this case, which incorporates a pseudomarginal-type
construction \citep{Andrieu2009}.

\paragraph{Data point tempering\label{sub:Data-point-tempering}}

For the SMC approach to be efficient we require that the sequence of distributions
$\{\pi_{t}\}$ be chosen such that $\pi_{0}$ is easy to simulate
from, $\pi_{T}$ is the target of interest and the intermediate distributions provide
a ``route'' between them. For the applications in this paper we
found the data tempering approach of \citet{Chopin2002} to be particularly
useful. Suppose that the data $y$ consists of $N$ points, and that
$N$ is exactly divisible by $T$ for ease of exposition. We then
take $\pi_{0}(\theta|y)=p(\theta)$ and for $t=1,...T$
$\pi_{t}(\theta|y)=p(\theta)f_{t}(y|\theta)$ with 
\begin{equation}
f_{t}(y|\theta)=f\left(y_{1:Nt/T}|\theta\right),
\end{equation}
i.e. essentially we incorporate $N/T$ additional data points for each increment of
$t$. On this sequence of targets we then propose to use the SMC sampler
with an MCMC kernel as described in the previous section. The only
slightly non-standard point is the estimation of $Z_{t-1}(\theta_{t-1}^{(p)})/$ $Z_{t}(\theta_{t-1}^{(p)})$,
since in this case $Z_{t-1}(\theta_{t-1}^{(p)})$
and $Z_{t}(\theta_{t-1}^{(p)})$ are the normalising constants of
distributions on different spaces. We use 
{\small
\begin{equation}
\widehat{\frac{Z_{t-1}(\theta_{t-1}^{(p)})}{Z_{t}(\theta_{t-1}^{(p)})}}=\frac{1}{M}\sum_{m=1}^{M}\frac{\gamma_{t-1}(v_{t}^{(m,p)}|\theta_{t-1}^{(p)})q_{w}(w_{t}^{(m,p)})}{\gamma_{t}(u_{t}^{(m,p)}|\theta_{t-1}^{(p)})}\label{eq:SMC_Z_tempering-1}
\end{equation}}
where $u_{t}^{(m,p)}\sim f_{t}(\cdot|\theta_{t-1}^{(p)})$ and $v_{t}^{(m,p)}$
and $w_{t}^{(m,p)}$ are subvectors of $u_{t}^{(m,p)}$. $w_{t}^{(m,p)}$
is in the space of the additional variables added when moving from
$f_{t-1}$ to $f_{t}$ (providing the argument in an arbitrary auxiliary
distribution $q_{w}(\cdot)$) and $v_{t}^{(m,p)}$ is in the space
of the existing variables. For $t=1$ this becomes 
\begin{equation}
\widehat{\frac{1}{Z_{1}(\theta_{0}^{(p)})}}=\frac{1}{M}\sum_{m=1}^{M}\frac{q_{w}(u_{1}^{(m,p)})}{\gamma_{1}(u_{1}^{(m,p)}|\theta_{0}^{(p)})} \label{eq:SMC_Z_tempering_step1}
\end{equation}
with $u_{1}^{(m,p)}\sim f_{t}(.|\theta_{0}^{(p)})$.

Analogous to the SAV method, a sensible choice for $q_{w}(w)$ might
be to use $f\left(w|\widehat{\theta}\right)$, where $w$ is on the
same space as \foreignlanguage{british}{$N/T$ data points. The normalising
constant for this distribution needs to be known to calculate the
importance weight in \eqref{eq:SMC_Z_tempering} so, as earlier,
we advocate estimating this in advance of running the SMC sampler
(aside from when the data points are added one at a time - in this
case the normalising constant may usually be found analytically).
Note that if $y$ does not consist of i.i.d. points, it is useful
to choose the order in which data points are added such that the same
$q_{w}$ (each with the same normalising constant) can be used in
every weight update. For example, in an Ising model, the requirement
would be to add the same shape grid of variables at each target.}

\paragraph{Marginal SMC}
An alternative method commonly used in ABC applications arises from the use of an approximation
to the optimal backward kernel \citep{Peters2005,Klaas2005}. In this
case the weight update is 
\begin{eqnarray}
\widetilde{w}_{t}^{(p)} 
 & = & \frac{p(\theta_{t}^{(p)})\gamma_{t}(y|\theta_{t}^{(p)})}{Z_{t}(\theta_{t}^{(p)})\sum_{r=1}^{P}w_{t-1}^{(r)}K_{t}(\theta_{t}^{(p)}|\theta_{t-1}^{(r)})}\label{eq:L_opt_back_Z}
\end{eqnarray}
for an arbitrary forward kernel $K_{t}$. 
This results in
a computational complexity of $O(P^{2})$ 
compared to $O(P)$ for a standard
SMC method, but we include it here in order to note that  the $1/Z(\cdot)$ term in \eqref{eq:L_opt_back_Z}
could be dealt with in the same way as in the simple IS case.
%
Considering the SAVM posterior, where in target $t$ we use the
distribution $q_{u}$ for the auxiliary variable $u_{t}$, and the
SAVM proposal, where $u_{t}^{(p)}\sim f_{t}(\cdot|\theta_{t}^{(p)})$ we arrive at
the weight update:
\begin{eqnarray*}
\widetilde{w}_{t}^{(p)} 
 & = & \frac{q_{u}(u_{t}^{(p)}|\theta_{t}^{(p)},y)p(\theta_{t}^{(p)})\gamma_{t}(y|\theta_{t}^{(p)})}{\gamma_{t}(u_{t}^{(p)}|\theta_{t}^{(p)})\sum_{r=1}^{P}w_{t-1}^{(r)}K_{t}(\theta_{t}^{(p)}|\theta_{t-1}^{(r)})}.
\end{eqnarray*}
in which normalising constant appears in this weight update. We include this
approach for completeness but do not investigate it further in this paper.

\subsubsection{Simulation-based SMC samplers\label{sub:Simulation-based-SMC-samplers}}
Section \ref{sub:Simulation-based-methods} describes how the ABC
and SL approximations may be used within IS. The same approximate
likelihoods may be used in SMC. In ABC \citep{Sisson2007d}, where
the sequence of targets is chosen to be $\pi_{t}(\theta)\propto p(\theta)\widehat{f}_{\epsilon_{t}}(y|\theta)$
with a decreasing sequence $\epsilon_{t}$, this idea provides a useful alternative to MCMC for exploring ABC posterior
distributions, whilst also providing estimates of Bayes' factors \citep{Didelot2011}.
The use of SMC with SL does not appear to have been explored previously.
One might expect SMC to be useful in this context (using, for example,
the sequence of targets $\pi_{t}(\theta)\propto p(\theta)\widehat{f}_{\mbox{SL}}^{(t/T)}\left(S(y)|\theta\right)$),
particularly when $\widehat{f}_{\mbox{SL}}$ is concentrated relative to the prior.

\subsection{Application to precision matrices\label{sub:Application-to-precision}}
In this section we examine the performance of the SMC sampler, with MCMC
proposal and data-tempered target distributions, for estimating the evidence in
an example in which $\theta$ is of moderately high dimension. We consider the case
in which $\theta=\Sigma^{-1}$ is an unknown precision matrix, $f(y|\theta)$ is the
$d$-dimensional multivariate Gaussian distribution with zero mean and $p(\theta)$
is a Wishart distribution $\mathcal{W}(\nu,V)$ with parameters $\nu\geq d$ and $V\in\mathbb{R}^{d\times d}$. Suppose
we observe $n$ \foreignlanguage{english}{i.i.d.} observations \foreignlanguage{english}{$y=\left\{ y_{i}\right\} _{i=1}^{n}$},
where $y_{i}\in\mathbb{R}^{d}$\foreignlanguage{english}{.} The true
evidence can be calculated analytically, and is given by 
\begin{equation}
p(y)=\frac{1}{\pi{}^{nd/2}}\frac{\Gamma_{d}(\frac{\nu+n}{2})}{\Gamma_{d}(\frac{\nu}{2})}\frac{\left|\left(V^{-1}+\sum_{i=1}^{n}y_{i}y_{i}^{T}\right)^{-1}\right|^{\frac{\nu+n}{2}}}{\left|V\right|^{\frac{\nu}{2}}},
\end{equation}
where $\Gamma_{d}$ denotes the $d$-dimensional gamma function. For
ease of implementation, we parametrise the precision using a Cholesky
decomposition $\Sigma^{-1}=LL'$ with $L$ a lower triangular matrix
whose $(i,j)$'th element is denoted $a_{ij}$.

As in section \ref{sub:Toy-example}, we write $f(y|\theta)$ as $\gamma(y|\theta)/Z(\theta)$
as follows
\selectlanguage{english}%
\begin{eqnarray}
f\left(\left\{ y_{i}\right\} _{i=1}^{n}\mid\Sigma^{-1}\right) 
 & = & \left|2\pi\Sigma\right|^{-n/2}\exp\left(-\frac{1}{2}\sum_{i=1}^{n}y_{i}'\Sigma^{-1}y_{i}\right),\nonumber\label{eq:mvn_llhd}
\end{eqnarray}
where in some of the experiments that follow,
$Z(\theta)=\left|2\pi \Sigma\right|^{n/2}$
is treated as if it is an INC. In the Wishart
prior, we take $\nu=10+d$ and $V=I_{d}$.

Taking $d=10$, $n=30$ points were simulated using $y_{i}\sim\mathcal{MVN}\left(0_{d},0.1\times I_{d}\right)$.
The parameter space is thus 55-dimensional, motivating the use of
an SMC sampler in place of IS or the population exchange method, neither
of which are suited to this problem. In the SMC sampler, in which we
used $P=10,000$ particles, the sequence of targets is given by data
point tempering. Specifically, the sequence of targets is to use $p(\Sigma^{-1})$
when $t=0$ and $p(\Sigma^{-1})f\left(\left\{ y_{i}\right\} _{i=1}^{t}\mid\Sigma^{-1}\right)$
for $t=1,...,T$ (with $T=n$). The parameters are $\left\{ a_{ij}\mid1\leq j\leq i\leq d\right\} $.
We use single component MH kernels to update each of the parameters,
with one (deterministic) sweep consisting of an update of each in
turn. For each $a_{ij}$ we use a Gaussian random walk proposal, where
at target $t$, the variance for the proposal used for $a_{ij}$ is
taken to be the sample variance of $a_{ij}$ at target $t-1$. For
updating the weights of each particle we used equation \ref{eq:SMC_Z_tempering-1},
where we chose $q_{w}(\cdot)=f\left(\cdot\mid\widehat{\Sigma^{-1}}\right)$
with $\widehat{\Sigma^{-1}}$ the maximum likelihood estimate of the
precision $\Sigma^{-1}$, and chose $M=200$ ``internal'' importance
sampling points. Systematic resampling was performed when the effective
sample size (ESS) fell below $P/2$.

We estimated the evidence 10 times using the SMC sampler and compared
the statistical properties of each algorithm using these estimates.
For our simulated data, the $\log$ of the true evidence was $-89.43$.
Over the 10 runs of the SMC sampler a five-number summary of the $\log$
evidence estimates was ($-90.01$, $-89.51$, $-89.35$, $-88.92$, $-88.37$).



\subsection{Application to Ising models}

In this section we apply the random weight SMC sampler to the Ising model data
considered in section \ref{sub:10by10}. We use SMC to estimate the marginal
likelihood of both the first and second order Ising models, then take the
ratio of these estimates to estimate the Bayes' factor. Note that in this case
the size of the parameter space is much smaller than in the precision example,
with the models having parameter spaces of sizes 1 and 2 respectively. The
excellent results achieved by IS in section \ref{sub:10by10} might seem to imply that SMC samplers are not required for this problem, but recall that we required preliminary runs of the exchange algorithm in order to design an appropriate importance proposal, along with an SMC sampler in order to estimate the normalising constant $Z(\widehat{\theta})$ of the distribution $q_u$ used for the auxiliary variables $u^{(m)}$. An SMC sampler offers a cleaner approach that requires less user tuning.

We applied the random weight SMC sampler described in section
\ref{sub:Random-weight-SMC}, with 500 particles, data point tempering (adding
one pixel at a time, taking $q_w$ to be $\mbox{Bern}(0.5)$), and using the
estimate of the ratio of normalising constants in the weight update from
equation \eqref{eq:SMC_Z_tempering-1} with $M=20$ importance points. Each of
these estimates requires simulating a single point from $\gamma_{t}(\cdot |
\theta_{t-1}^{(p)})$ using a Gibbs sampler, which had a burn in of $B=10$
iterations, yielding a total computational budget of 200 Gibbs sweeps for
estimating the ratio of normalising constants. Note that, as considered in
section \ref{sub:10by10}, this use of a Gibbs sampler results in an inexact
algorithm, but this level of burn in was found to be sufficient for this bias
to be minimal in the random weight IS algorithms. The MCMC kernel of the
exchange algorithm was used (with proposal taken to be the sample variance of the
set of particles at each SMC iteration), using the approximate version where a
Gibbs sampler with burn in $B=10$ iterations is used to simulate from $\gamma_{t}(\cdot
| \theta^{(*)})$. The total cost of this algorithm is comparable to the IS
approaches in section \ref{sub:10by10}, with a total cost of $5.25 \times 10^6$
Gibbs sweeps and hence around a quarter of that of the algorithm of \citet{Friel2013e}. Figure \ref{fig:Box-plots-SMC} shows the results produced
by this method in comparison with those from \citet{Friel2013e}.

\begin{figure}
\includegraphics[scale=0.30]{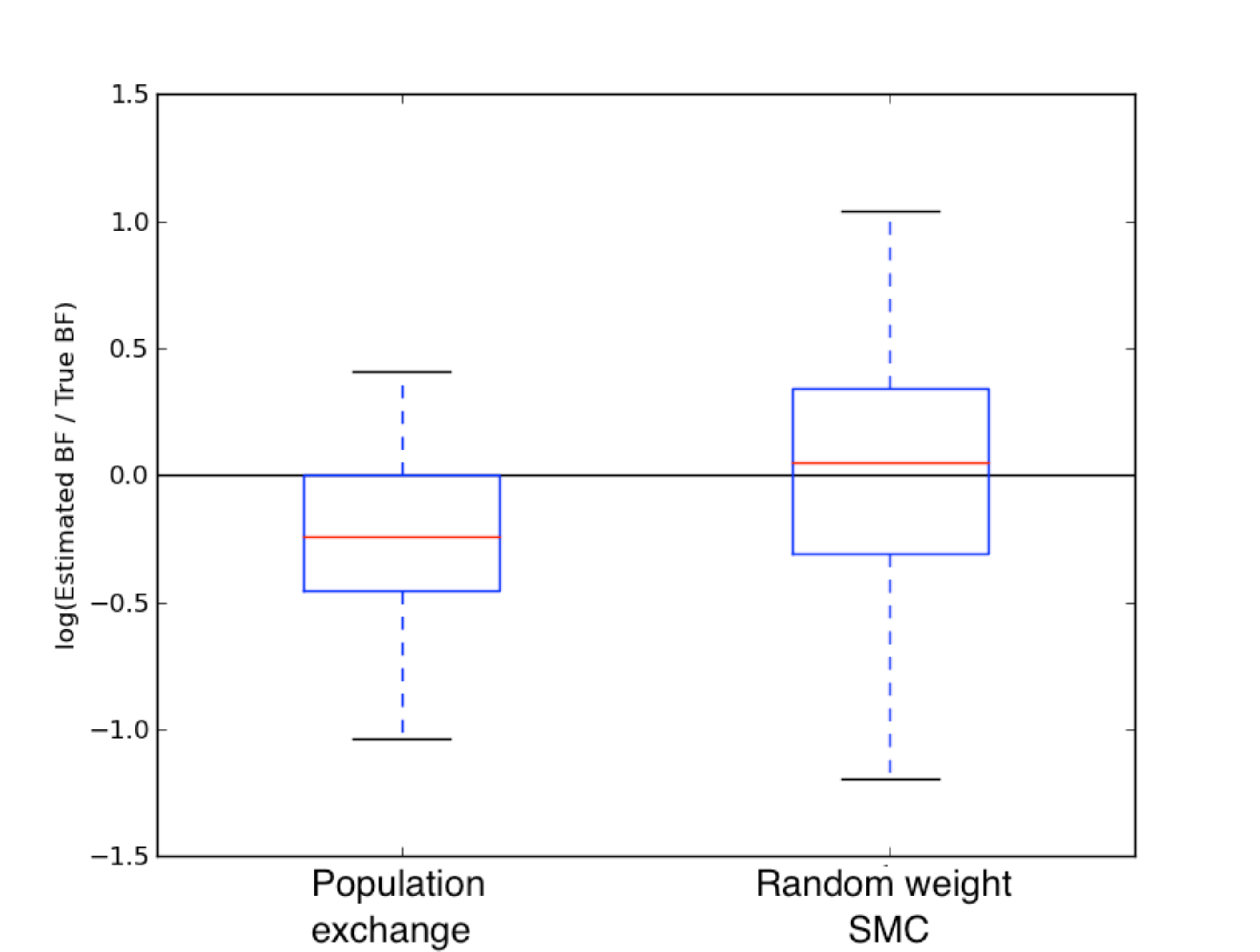}

\protect\caption{Box plots of the results of population exchange and random weight SMC.\label{fig:Box-plots-SMC}}
\end{figure}

We observe that the median of the random weight SMC estimates is more accurate than that of the population exchange estimates - the bias introduced through using an internal Gibbs sampler in place of an exact sampler does not appear to accumulate sufficiently to affect the results (this issue is explored further in the following section). However, it has slightly higher variance than population exchange (much higher than SAVIS and MAVIS). This high variance can be attributed to two factors:

\begin{enumerate}
\item Since the SMC sampler begins with points sampled from the prior, larger changes in $\theta$ are considered than in the IS approaches, thus the estimates of the ratio of the normalising constants require more importance points to be accurate - the results suggest that the budget of 200 Gibbs sweeps is insufficient. This is the opposite situation to that encountered in section \ref{sub:100by100}, where the changes in $\theta$ are small and the estimates of the ratio of the normalising constants are accurate with small numbers of importance points.
\item It's been frequently observed (cf. \citet{Lee2015})  that, as suggested
  by the asymptotic variance expansion, in some instances the first few iterations of an SMC sampler contribute substantially to the ultimate error. This issue arises since the forgetting of the sampler doesn't suppress the terms that the initial errors contribute to the asymptotic variance enough to compensate for the fact that they're much larger than the final ones. This is due, when using data point tempering in the manner we have here, to the much larger relative discrepancy between the first few distributions in the sequence than between later distributions.
\end{enumerate}

We conclude that the random weight SMC method is a viable approach to estimating Bayes' factors for these models, but that care should be taken in tuning the weight estimates and choosing the sequence of SMC distributions.

\subsection{Biased Weights in SMC}\label{sec:biased-smc}
\subsubsection{Error bounds}
We now examine the effect of using inexact weights on estimates produced
by SMC samplers. By way of theoretical motivation of such an approach,
we demonstrate that under strong, but standard (cf. \citet{DelMoral2004}), assumptions on the mixing of the
sampler, if the approximation error is sufficiently small, then this
error can be controlled uniformly over the iterations of the algorithm
and will \emph{not} accumulate unboundedly over time (and that it
can in principle be made arbitrarily small by making the relative
bias small enough for the desired level of accuracy). We do not here
consider the particle system itself, but rather the sequence of distributions
which are being approximated by Monte Carlo in the approximate version
of the algorithm and in the idealised algorithm being approximated.
The Monte Carlo approximation of this sequence can then be understood
as a simple mean field approximation and its convergence has been
well studied, see for example \citet{DelMoral2004}.

In order to do this, we make a number of identifications in order
to allow the consideration of the approximation in an abstract manner.
We allow $\widetilde{G}_{t}$ to denote the incremental weight function
at time $t$, and $G_{t}$ to denote the \emph{exact} weight function
which it approximates (any auxiliary random variables needed in order
to obtain this approximation are simply added to the state space and
their sampling distribution to the transition kernel). The transition
kernel $M_{t}$ combines the proposal distribution of the SMC algorithm
together with the sampling distribution of any needed auxiliary variables.
We allow $x$ to denote the full collection of variables sampled during
an iteration of the sampler, which is assumed to exist on the same
space during each iteration of the sampler.

We employ the following assumptions (we assume an infinite sequence
of algorithm steps and associated target distributions, proposals
and importance weights; naturally, in practice only a finite number
would be employed but this formalism allows for a straightforward
statement of the result): 
\begin{description}
\item [{A1}] (Bounded Relative Approximation Error) There exists $\gamma<\infty$
such that: 
\[
\sup_{t\in\mathbb{N}}\sup_{x}\frac{|G_{t}(x)-\widetilde{G}_{t}(x)|}{\widetilde{G}_{t}(x)}\leq\gamma.
\]

\item [{A2}] (Strong Mixing; slightly stronger than a global Doeblin condition)
There exists \foreignlanguage{english}{$\epsilon(M)>0$ such that:}
\[
\sup_{t\in\mathbb{{N}}}\inf_{x,y}\frac{dM_{t}(x,\cdot)}{dM_{t}(y,\cdot)}\geq\epsilon(M).
\]

\item [{A3}] (Control of Potential) There exists \foreignlanguage{english}{$\epsilon(G)>0$
such that:} 
\[
\sup_{t\in\mathbb{{N}}}\inf_{x,y}\frac{G_{t}(x)}{G_{t}(y)}\geq\epsilon(G).
\]

\end{description}
The first of these assumptions controls the error introduced by employing
an inexact weighting function; the others ensure that the underlying
dynamic system is sufficiently ergodic to forget it's initial conditions
and hence limit the accumulation of errors. We demonstrate below that
the combination of these properties suffices to transfer that stability
to the approximating system.

We consider the behaviour of the distributions $\eta_{p}$ and $\tilde{\eta}_{p}$
which correspond to the target distributions at iteration $p$ of
the exact and approximating algorithms, prior to reweighting, at iteration
$p$ in the following proposition, the proof of which is provided
in Appendix \ref{sec:ubound_proof}, which demonstrates that if the
approximation error, $\gamma$, is sufficiently small then the accumulation
of error over time is controlled:
\begin{prop}[Uniform Bound on Total-Variation Discrepancy]
\label{th:ubound}If A1, A2 and A3 hold then:
\[
\sup_{n\in\mathbb{N}}\left\Vert \eta_{n}-\widetilde{\eta}_{n}\right\Vert _{\mbox{TV}}\leq\frac{4\gamma(1-\epsilon(M))}{\epsilon^{3}(M)\epsilon(G)}.
\]

\end{prop}
This result is not intended to do any more than demonstrate that,
qualitatively, such forgetting can prevent the accumulation of error
even in systems with ``biased'' importance weighting potentials.
In practice, one would wish to make use of more sophisticated ergodicity
results such as those of \citet{Whiteley2013}, within this framework
to obtain results which are somewhat more broadly applicable: assumptions
A2 and A3 are very strong, and are used only because they allow stability
to be established simply. Although this result is, in isolation, too
weak to justify the use of the approximation schemes introduced here
in practice, together with the empirical results presented below, it does
suggest that further investigation of such approximations is warranted
particularly in settings in which unbiased estimators are not available.

\subsubsection{Empirical results}

We use the precision example introduced in section \ref{sub:Application-to-precision}
to investigate the effect of using biased weights in SMC samplers.
Specifically we take $d=1$ and use a simulated dataset $y$ where
$n=5000$ points were simulated using $y_{i}\sim\mathcal{N}\left(0,0.1\right)$.
In this case there is only a single parameter to estimate, $a_{1}$,
and we examine the bias of estimates of the evidence using four alternative
SMC samplers, each of which use a data-tempered sequence of targets
(adding one data point at each target). For this data we can calculate analytically the true value of the marginal likelihood after receiving each data point, thus we can estimate the bias of each sampler at each iteration. The first SMC sampler (the ``exact weight'' sampler) is the method where
the true value of $Z_{t-1}(\theta_{t-1}^{(p)})/Z_{t}(\theta_{t-1}^{(p)})$ is used in the weight update. The second is the same ``unbiased random weight'' sampler used in section \ref{sub:Application-to-precision},
which uses an unbiased IS weight estimate, here with $M=20$ ``internal''
importance sampling points. The third, which we refer to as the ``biased random weight'' sampler,
uses a biased bridge estimator instead, specifically we use in place
of \eqref{eq:SMC_Z_tempering-1}
\begin{equation}
\widehat{\frac{Z_{t-1}(\theta_{t-1}^{(p)})}{Z_{t}(\theta_{t-1}^{(p)})}}=\left(\sum_{m=1}^{M/2}\left[\frac{\gamma_{t-1}(v_{t,1}^{(m,p)}|\theta_{t-1}^{(p)})q_{w}(w_{t,1}^{(m,p)})}{\gamma_{t}(u_{t,1}^{(m,p)}|\theta_{t-1}^{(p)})}\right]^{1/2}\right)\left/ \right.\nonumber
\end{equation}
\begin{equation}
\qquad \left(\sum_{m=1}^{M/2}\left[\frac{\gamma_{t}(u_{t,2}^{(m,p)}|\theta_{t-1}^{(p)})}{\gamma_{t-1}(v_{t,2}^{(m,p)}|\theta_{t-1}^{(p)})q_{w}(w_{t,2}^{(m,p)})}\right]^{1/2}\right),\label{eq:SMC_Z_tempering}
\end{equation}
where $v_{t,2}^{(m,p)}\sim f_{t-1}(.|\theta_{t-1}^{(p)})$, $w_{t,2}^{(m,p)}\sim q_{w}(.)$
so that $u_{t,2}^{(m,p)}=\left(v_{t,2}^{(m,p)},w_{t,2}^{(m,p)}\right)$,
and $u_{t,1}^{(m,p)}\sim f_{t}(.|\theta_{t-1}^{(p)})$ with $v_{t,1}^{(m,p)}$
and $w_{t,1}^{(m,p)}$ being the corresponding subvectors of $u_{t,1}^{(m,p)}$.

Motivated by the theoretical argument presented previously,
we investigate the effect of improving the mixing of the kernel used
within the SMC. In this model the exact posterior is available at
each SMC target, so we may replace the use of an MCMC move to update
the parameter with a direct simulation from the posterior. In this
extreme case, there is no dependence between each particle and its
history; we refer to this, the fourth SMC sampler we consider, as ``biased random weight with perfect mixing''.  Each SMC sampler was run 20 times, using
50 particles.

Figures \ref{fig:smc_bias} and \ref{fig:smc_mse} show the estimated bias and mean square error of the $\log$ evidence estimates of each sampler at each iteration\footnote{We note that log of an unbiased estimate in fact produces a negatively-biased estimator but we observe, through the results for the exact algorithm indicate that the variance of the evidence estimates we use is sufficiently small that this effect is negligible.}. No bias is observed in the algorithm with true weights, and only a small bias is observed in the unbiased random weight sampler (this bias is likely to be due to the relatively small number of replications). Bias does accumulate in the biased random weight sampler, but we note that the level of bias appears to stabilise. This accumulation of bias means that one should exercise caution in the use of SMC samplers with biased weights. However, we observe that perfect mixing substantially decreases the bias in the evidence estimates from the algorithm. Also, in this case we observe that the bias does not accumulate sufficiently to give poor estimates of the evidence. Here the standard deviation of the final $\log$ evidence estimate over the random weight SMC sampler runs is approximately 0.4, so the bias is not large by comparison. 

\begin{figure}
\centering{}\includegraphics[scale=0.45]{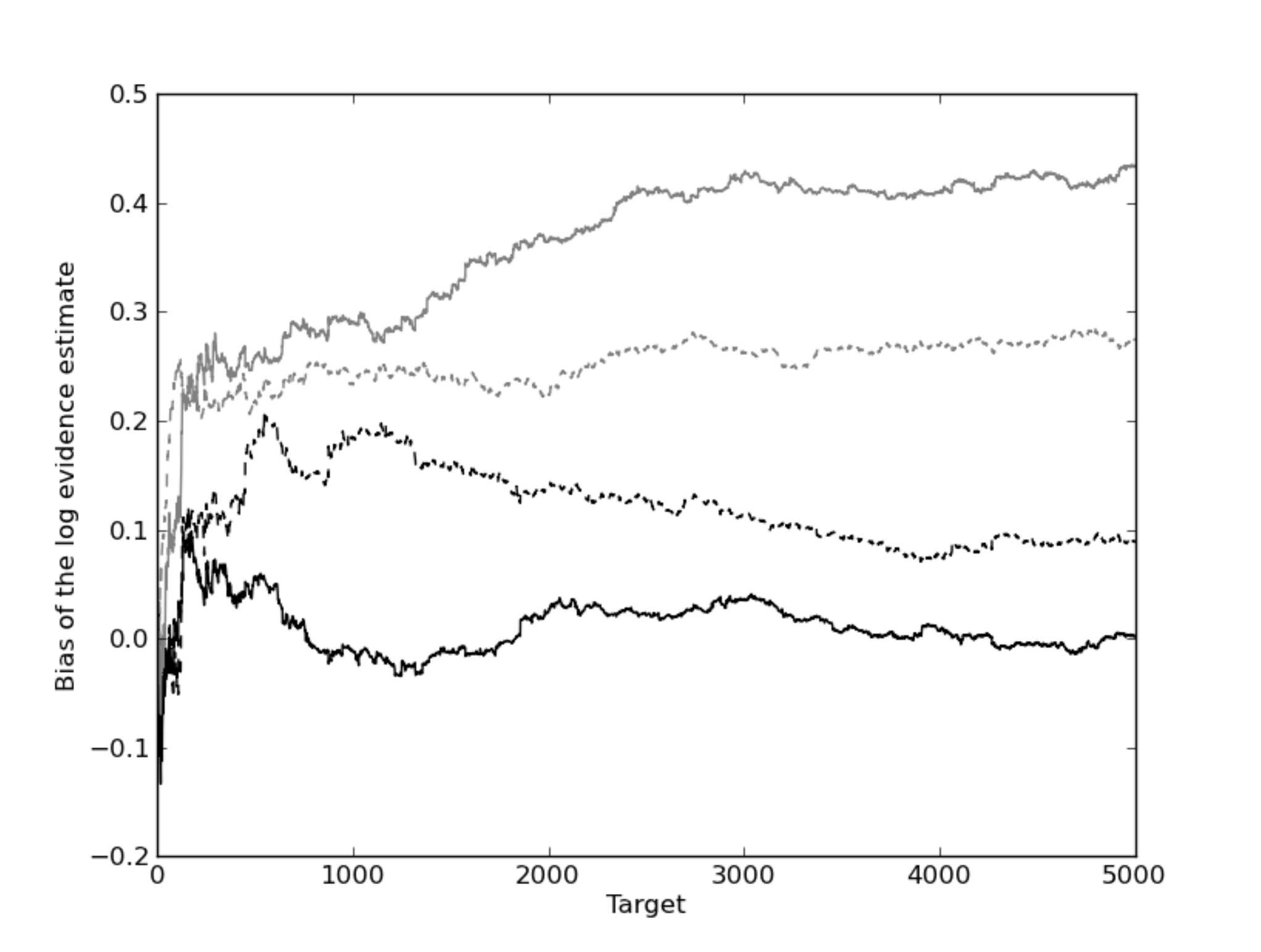}\protect\caption{The estimated bias in the $\log$ evidence estimates of the true (black solid), unbiased random weight (black dashed), biased random weight (grey solid) SMC algorithms using MCMC kernels, and the estimated bias when using the biased random weight algorithm with perfect mixing (grey dashed).
\label{fig:smc_bias}}
\end{figure}

\begin{figure}
\centering{}\includegraphics[scale=0.45]{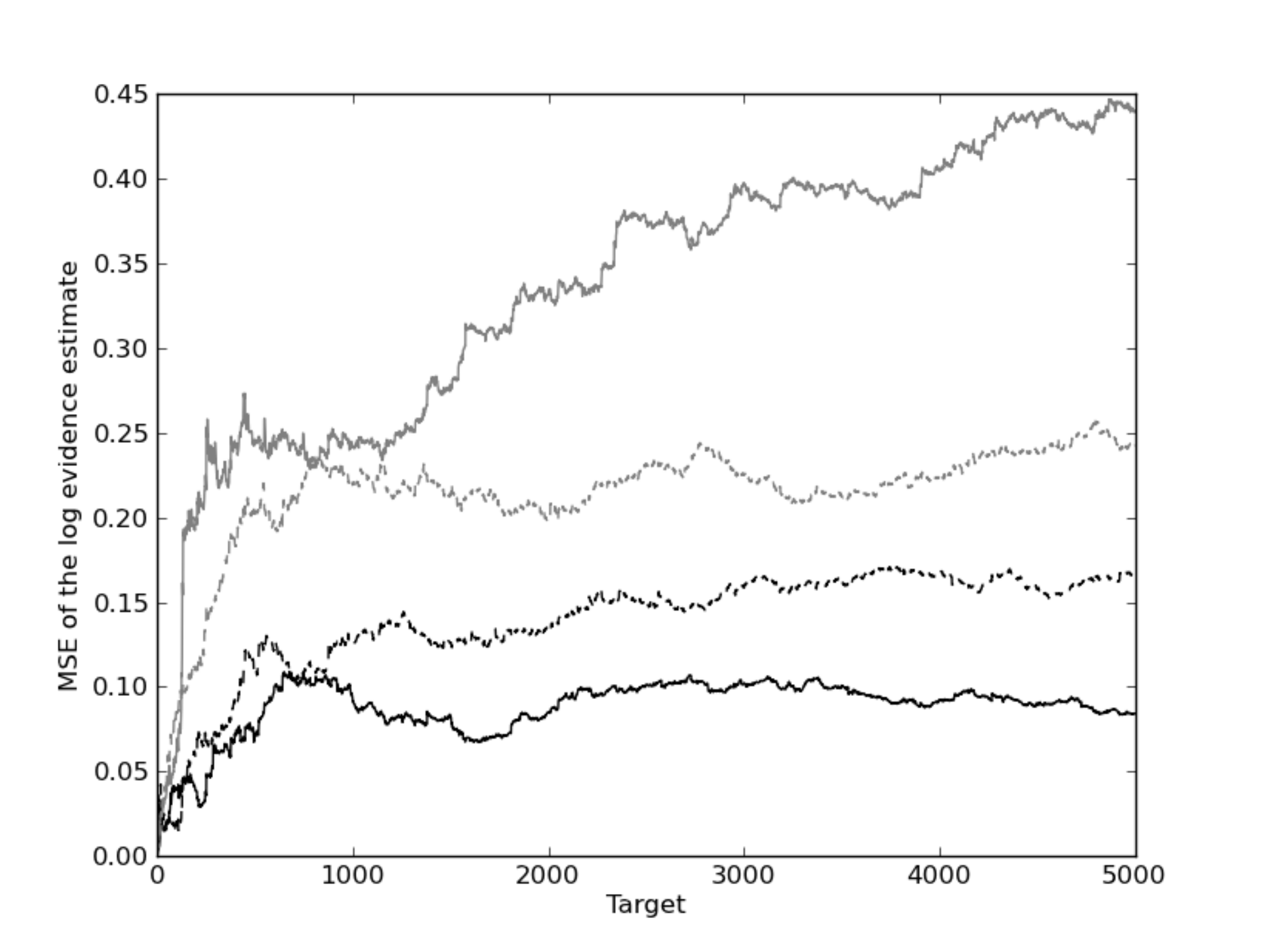}\protect\caption{The estimated MSE in the $\log$ evidence estimates of the four SMC samplers (same key as figure \ref{fig:smc_bias}).
\label{fig:smc_mse}}
\end{figure}

\subsection{Discussion}

In section \ref{sub:Application-to-Ising} we observed clearly that
the use of biased weights in IS can be useful for estimating the evidence
in doubly intractable models, but we have not observed the same for
SMC with biased weights. When applied to the precision example in
section \ref{sub:Application-to-precision}, an inexact sampler (using
the bridge estimator) did not outperform the exact sampler, despite
the mean square error of the biased bridge weight estimates being
substantially improved compared to the unbiased IS estimate. Over
10 runs the mean square error in the $\log$ evidence was 0.34 for
the inexact sampler, compared to 0.28 for the exact sampler. This
experience suggests that samplers with biased weights should be used
with caution: weight estimates with low variance do not guarantee
good performance due to the accumulation of bias in the SMC.

However, the theoretical and empirical investigation in this section
suggests that this idea is worth further investigation, possibly for
situations involving some of the other intractable likelihoods listed
in section \ref{sec:Introduction}. Our results suggest that improved
mixing can help combat the accumulation of bias, which may imply
that there may be situations where it is useful to perform many iterations
of a kernel at a particular target, rather than the more standard
approach of using many intermediate targets at each of which a single
iteration of a kernel is used. Other variations are also possible,
such as the calculation of fast cheap biased weights at each target
in order only to adaptively decide when to resample, with more accurate
weight estimates (to ensure accurate resampling and accurate estimates
based on the particles) only calculated when the method chooses to
resample.

\section{Conclusions\label{sec:Conclusions}}

This paper describes several IS and SMC approaches for estimating
the evidence in models with INCs that outperform previously described
approaches. These methods may also prove to be useful alternatives
to MCMC for parameter estimation. Several of the ideas in the paper
are also applicable more generally, in particular the use of synthetic
likelihood in the IS context and the notion of using biased weight
estimates in IS and SMC. We note that the bias in these biased weight methods may be small compared to errors resulting from commonly accepted approximate techniques such as ABC.

For biased IS, in section \ref{sub:IS-with-biased} we show that the error of
estimates from low-variance biased methods can be less than those from
unbiased methods of higher variance. This is comparable to a result for biased
MCMC methods \citep{Johndrow2015}, where it is shown that the error of
estimates from a computationally cheap biased MCMC can be less than those from
an expensive exact MCMC. In both cases, given a finite computational budget,
it is not always the case that this budget should be spent on guaranteeing the
exactness of the sampler if minimizing approximation error is the objective. 

A similar choice concerning the allocation of computational resources arises in
SMC. Here, one does need to be especially careful about the use of biased
SMC, due to the possible accumulation of bias over SMC iterations. One
might expect this accumulated bias to outweigh any benefits a reduced variance
may bring. For this reason we advise caution in the use of biased SMC in
general. This paper does, however, indicate that there may exist cases where
biased SMC is useful, through: the theoretical result that under strong mixing
conditions bias does not accumulate unboundedly; the empirical evidence that
fast mixing may reduce the accumulation of bias; and the empirical results
where we observe (in a situation where the distance between successive targets
decreases) that the rate at which bias accumulates decreases with time. 

\vspace{10 mm}

\renewcommand{\abstractname}{Acknowledgements}
\begin{abstract}
The authors would like to thank Nial Friel for useful discussions, and for giving us access to the data and results from \citet{Friel2013e}.
\end{abstract}


\begin{thebibliography}{}
%
%
\bibitem[{Alquier et~al(2015)Alquier, Friel, Everitt, and Boland}]{Alquier2014}
Alquier P, Friel N, Everitt RG, Boland A (2015) {Noisy Monte Carlo: Convergence
  of Markov chains with approximate transition kernels}. Statistics and
  Computing In press.

\bibitem[{Andrieu and Roberts(2009)}]{Andrieu2009}
Andrieu C, Roberts GO (2009) {The pseudo-marginal approach for efficient Monte
  Carlo computations}. The Annals of Statistics 37(2):697--725

\bibitem[{Andrieu and Vihola(2012)}]{Andrieu2012d}
Andrieu C, Vihola M (2012) {Convergence properties of pseudo-marginal Markov
  chain Monte Carlo algorithms}. arXiv (1210.1484)

\bibitem[{Beaumont(2003)}]{Beaumont2003}
Beaumont MA (2003) {Estimation of population growth or decline in genetically
  monitored populations.} Genetics 164(3):1139--1160

\bibitem[{Beskos et~al(2011)Beskos, Crisan, Jasra, and Whiteley}]{Beskos2011c}
Beskos A, Crisan D, Jasra A, Whiteley N (2011) {Error Bounds and Normalizing
  Constants for Sequential Monte Carlo in High Dimensions}. arXiv (1112.1544)

\bibitem[{Caimo and Friel(2011)}]{Caimo2011}
Caimo A, Friel N (2011) {Bayesian inference for exponential random graph
  models}. Social Networks 33:41--55

\bibitem[{Chopin(2002)}]{Chopin2002}
Chopin N (2002) {A sequential particle filter method for static models}.
  Biometrika 89(3):539--552

\bibitem[{Chopin et~al(2013)Chopin, Jacob, and Papaspiliopoulos}]{Chopin2013}
Chopin N, Jacob PE, Papaspiliopoulos O (2013) {$\textrm{SMC}^2$: an efficient
  algorithm for sequential analysis of state space models}. Journal of the
  Royal Statistical Society: Series B 75(3):397--426

\bibitem[{Del~Moral(2004)}]{DelMoral2004}
Del~Moral P (2004) {Feynman-Kac} formulae: genealogical and interacting
  particle systems with applications. Probability and Its Applications,
  Springer, New York

\bibitem[{{Del Moral} et~al(2006){Del Moral}, Doucet, and
  Jasra}]{DelMoral2006c}
{Del Moral} P, Doucet A, Jasra A (2006) {Sequential Monte Carlo samplers}.
  Journal of the Royal Statistical Society: Series B 68(3):411--436

\bibitem[{{Del Moral} et~al(2007){Del Moral}, Doucet, and
  Jasra}]{DelMoral2007a}
{Del Moral} P, Doucet A, Jasra A (2007) {Sequential Monte Carlo for Bayesian
  computation}. Bayesian Statistics 8:115--148

\bibitem[{Didelot et~al(2011)Didelot, Everitt, Johansen, and
  Lawson}]{Didelot2011}
Didelot X, Everitt RG, Johansen AM, Lawson DJ (2011) {Likelihood-free
  estimation of model evidence}. Bayesian Analysis 6(1):49--76

\bibitem[{Drovandi et~al(2015)Drovandi, Pettitt, and Lee}]{Drovandi2013}
Drovandi CC, Pettitt AN, Lee A (2015) Bayesian indirect inference using a
  parametric auxiliary model. Statistical Science 30(1):72--95

\bibitem[{Everitt(2012)}]{Everitt2012}
Everitt RG (2012) {Bayesian Parameter Estimation for Latent Markov Random
  Fields and Social Networks}. Journal of Computational and Graphical
  Statistics 21(4):940--960

\bibitem[{Fearnhead et~al(2010)Fearnhead, Papaspiliopoulos, Roberts, and
  Stuart}]{Fearnhead2010f}
Fearnhead P, Papaspiliopoulos O, Roberts GO, Stuart AM (2010) {Random-weight
  particle filtering of continuous time processes}. Journal of the Royal
  Statistical Society Series B 72(4):497--512

\bibitem[{Friel(2013)}]{Friel2013e}
Friel N (2013) {Evidence and Bayes factor estimation for Gibbs random fields}.
  Journal of Computational and Graphical Statistics 22(3):518--532

\bibitem[{Friel and Rue(2007)}]{Friel2007}
Friel N, Rue H (2007) {Recursive computing and simulation-free inference for
  general factorizable models}. Biometrika 94(3):661--672

\bibitem[{Girolami et~al(2013)Girolami, Lyne, Strathmann, Simpson, and
  Atchade}]{Girolami2013}
Girolami MA, Lyne AM, Strathmann H, Simpson D, Atchade Y (2013) {Playing
  Russian Roulette with Intractable Likelihoods}. arXiv (1306.4032)

\bibitem[{Grelaud et~al(2009)Grelaud, Robert, and Marin}]{Grelaud2009}
Grelaud A, Robert CP, Marin JM (2009) {ABC likelihood-free methods for model
  choice in Gibbs random fields}. Bayesian Analysis 4(2):317--336

\bibitem[{Johndrow et~al(2015)Johndrow, Mattingly, Mukherjee and Dunson}]{Johndrow2015}
Johndrow JE, Mattingly JC, Mukherjee S and Dunson D (2015) {Approximations of Markov Chains and High-Dimensional Bayesian Inference}. arXiv (1508.03387)

\bibitem[{Klaas et~al(2005)Klaas, de~Freitas, and Doucet}]{Klaas2005}
Klaas M, de~Freitas N, Doucet A (2005) {Toward practical $N^2$ Monte Carlo: The
  marginal particle filter}. In: Proceedings of the 20th International
  Conference on Uncertainty in Artificial Intelligence

\bibitem[{Kong et~al(1994)Kong, Liu, and Wong}]{Kong1994}
Kong A, Liu JS, Wong WH (1994) Sequential imputations and {Bayesian} missing
  data problems. Journal of the American Statistical Association 89(425):278--288
  
\bibitem[{Lee and Whiteley(2015)Lee, and Whiteley}]{Lee2015}
Lee A, Whiteley N (2015) {Variance estimation and allocation in the
particle filter} arXiv (2015.0394)

\bibitem[{Marin et~al(2014)Marin, Pillai, Robert, and Rousseau}]{Marin2013}
Marin JM, Pillai NS, Robert CP, Rousseau J (2014) {Relevant statistics for
  Bayesian model choice}. Journal of the Royal Statistical Society: Series B
  (Statistical Methodology) 76(5):833--859

\bibitem[{Marjoram et~al(2003)Marjoram, Molitor, Plagnol, and
  Tavare}]{Marjoram2003a}
Marjoram P, Molitor J, Plagnol V, Tavare S (2003) {Markov chain Monte Carlo
  without likelihoods}. Proceedings of the National Academy of Sciences of the
  United States of America 100(26):15,324--15,328

\bibitem[{Meng and Wong(1996)}]{Meng1996c}
Meng Xl, Wong WH (1996) {Simulating ratios of normalizing constants via a
  simple identity: a theoretical exploration}. Statistica Sinica 6:831--860

\bibitem[{M{\o}ller et~al(2006)M{\o}ller, Pettitt, Reeves, and
  Berthelsen}]{Moller2006}
M{\o}ller J, Pettitt AN, Reeves RW, Berthelsen KK (2006) {An efficient Markov
  chain Monte Carlo method for distributions with intractable normalising
  constants}. Biometrika 93(2):451--458

\bibitem[{Murray et~al(2006)Murray, Ghahramani, and MacKay}]{Murray2006}
Murray I, Ghahramani Z, MacKay DJC (2006) {MCMC for doubly-intractable
  distributions}. In: Proceedings of the 22nd Annual Conference on Uncertainty
  in Artificial Intelligence (UAI), pp 359--366

\bibitem[{Neal(2001)}]{Neal2001}
Neal RM (2001) {Annealed importance sampling}. Statistics and Computing
  11(2):125--139

\bibitem[{Neal(2005)}]{Neal2005e}
Neal RM (2005) {Estimating Ratios of Normalizing Constants Using Linked
  Importance Sampling}. arXiv (0511.1216)

\bibitem[{Nicholls et~al(2012)Nicholls, Fox, and Watt}]{Nicholls2012}
Nicholls GK, Fox C, Watt AM (2012) {Coupled MCMC With A Randomized Acceptance
  Probability}. arXiv (1205.6857)

\bibitem[{Peters(2005)}]{Peters2005}
Peters GW (2005) {Topics in Sequential Monte Carlo Samplers}. {M.Sc.} thesis,
  Unviersity of Cambridge

\bibitem[{Picchini and Forman(2013)}]{Picchini2013}
Picchini U, Forman JL (2013) {Accelerating inference for diffusions observed
  with measurement error and large sample sizes using Approximate Bayesian
  Computation: A case study}. arXiv (1310.0973)

\bibitem[{Prangle et~al(2014)Prangle, Fearnhead, Cox, Biggs, and
  French}]{Prangle2013}
Prangle D, Fearnhead P, Cox MP, Biggs PJ, French NP (2014) {Semi-automatic
  selection of summary statistics for ABC model choice}. Statistical
  Applications in Genetics and Molecular Biology 13(1):67--82

\bibitem[{Rao et~al(2013)Rao, Lin, and Dunson}]{Rao2013}
Rao V, Lin L, Dunson DB (2013) {Bayesian inference on the Stiefel manifold}.
  arXiv (1311.0907)

\bibitem[{Robert et~al(2011)Robert, Cornuet, Marin, and Pillai}]{Robert2011j}
Robert CP, Cornuet JM, Marin JM, Pillai NS (2011) {Lack of confidence in
  approximate Bayesian computation model choice}. Proceedings of the National
  Academy of Sciences of the United States of America 108(37):15,112--7

\bibitem[{Schweinberger and Handcock(2015)}]{Schweinberger2015}
Schweinberger M, Handcock M (2015) {Local dependence in random graph models:
  characterization, properties and statistical inference}. Journal of the Royal
  Statistical Society: Series B In press.

\bibitem[{Sisson et~al(2007)Sisson, Fan, and Tanaka}]{Sisson2007d}
Sisson SA, Fan Y, Tanaka MM (2007) {Sequential Monte Carlo without
  likelihoods}. Proceedings of the National Academy of Sciences of the United
  States of America 104(6):1760--1765

\bibitem[{Skilling(2006)}]{Skilling2006}
Skilling J (2006) {Nested sampling for general Bayesian computation}. Bayesian
  Analysis 1(4):833--859

\bibitem[{Tavar\'{e} et~al(1997)Tavar\'{e}, Balding, Griffiths, and
  Donnelly}]{Tavare1997f}
Tavar\'{e} S, Balding DJ, Griffiths RC, Donnelly PJ (1997) {Inferring
  Coalescence Times From DNA Sequence Data}. Genetics 145(2):505--518

\bibitem[{Tran et~al(2013)Tran, Scharth, Pitt, and Kohn}]{Tran2013}
Tran MN, Scharth M, Pitt MK, Kohn R (2013) {$\textrm{IS}^2$} for Bayesian
  inference in latent variable models. arXiv (1309.3339)

\bibitem[{Whiteley(2013)}]{Whiteley2013}
Whiteley N (2013) {Stability properties of some particle filters}. Annals of
  Applied Probability 23(6):2500--2537

\bibitem[{Wilkinson(2013)}]{Wilkinson2008}
Wilkinson RD (2013) {Approximate Bayesian computation (ABC) gives exact results
  under the assumption of model error}. Statistical Applications in Genetics
  and Molecular Biology 12(2):129--141

\bibitem[{Wood(2010)}]{Wood2010f}
Wood SN (2010) {Statistical inference for noisy nonlinear ecological dynamic
  systems}. Nature 466(August):1102--1104

\bibitem[{Zhou et~al(2015)Zhou, Johansen, and Aston}]{Zhou2013}
Zhou Y, Johansen AM, Aston JAD (2015) Towards automatic model comparison: An
  adaptive sequential {M}onte {C}arlo approach. Journal of Computational and
  Graphical Statistics In press.
\end{thebibliography}

\appendix

\section{Using SAV and exchange MCMC within SMC\label{sec:Using-SAV-and}}

\subsection{Weight update when using SAV-MCMC}

\selectlanguage{english}%
Let us consider the SAVM posterior, with $K$ being the MCMC move
used in SAVM. In this case the weight update is 
\begin{eqnarray*}
\widetilde{w}_{k}^{(p)} & = & \frac{p(\theta_{t}^{(p)})f_{t}(y|\theta_{t}^{(p)})q_{u}(u_{t}^{(p)}|\theta_{t}^{(p)},y)}{p(\theta_{t-1}^{(p)})f_{t-1}(y|\theta_{t-1}^{(p)})q_{u}(u_{t-1}^{(p)}|\theta_{t-1}^{(p)},y)}\\
 & & \qquad \frac{L_{t-1}((\theta_{t}^{(p)},u_{t}^{(p)}),(\theta_{t-1}^{(p)},u_{t-1}^{(p)}))}{K_{t}((\theta_{t-1}^{(p)},u_{t-1}^{(p)}),(\theta_{t}^{(p)},u_{t}^{(p)}))}\\
 & = & \frac{p(\theta_{t}^{(p)})f_{t}(y|\theta_{t}^{(p)})q_{u}(u_{t}^{(p)}|\theta_{t}^{(p)},y)}{p(\theta_{t-1}^{(p)})f_{t-1}(y|\theta_{t-1}^{(p)})q_{u}(u_{t-1}^{(p)}|\theta_{t-1}^{(p)},y)}\\
 & & \qquad \frac{p(\theta_{t-1}^{(p)})f_{t}(y|\theta_{t-1}^{(p)})q_{u}(u_{t-1}^{(p)}|\theta_{t-1}^{(p)},y)}{p(\theta_{t}^{(p)})f_{t}(y|\theta_{t}^{(p)})q_{u}(u_{t}^{(p)}|\theta_{t}^{(p)},y)}\\
 & = & \frac{\gamma_{t}(y|\theta_{t-1}^{(p)})}{\gamma_{t-1}(y|\theta_{t-1}^{(p)})}\frac{Z_{t-1}(\theta_{t-1}^{(p)})}{Z_{t}(\theta_{t-1}^{(p)})},
\end{eqnarray*}
which is the same update as if we could use MCMC directly.

\selectlanguage{british}%

\subsection{Weight update when using the exchange algorithm}

\citet{Nicholls2012} show the exchange algorithm, when set up to
target $\pi_{t}(\theta|y)\propto p(\theta)f_{t}(y|\theta)$ in the
manner described in section \ref{sub:Exchange-algorithms}, simulates
a transition kernel that is in detailed balance with $\pi_{t}(\theta|y)$.
This follows from showing that it satisfies a ``very detailed balance''
condition, which takes account of the auxiliary variable $u$. The
result is that the derivation of the weight update follows exactly
that of \eqref{eq:smc-mcmc}.

\section{An extended space construction for the random weight SMC method in
section \ref{sub:SMC-with-an} \label{sec:An-Extended-Space}}

The following extended space construction justifies the use of the
``approximate'' weights in \eqref{eq:smc_unbiased_Z} via
an explicit sequential importance (re)sampling argument along the
lines of \citet{DelMoral2006c}, albeit with a slightly different
sequence of target distributions.

Consider an actual sequence of target distributions $\{\pi_{t}\}_{t\geq0}$.
Assume we seek to approximate a normalising constant during every
iteration by introducing additional variables $u_{t}=(u_{t}^{1},\ldots,u_{t}^{M})$
during iteration $t>0$.

Define the sequence of target distributions: 
\begin{align*}
 & \widetilde{\pi}_{t}\left(\widetilde{x}_{t}=(\theta_{0},\theta_{1},u_{1},\ldots,\theta_{t},u_{t})\right)\\
:= & \pi_{t}(\theta_{t})\prod_{s=0}^{t-1}L_{s}(\theta_{s+1},\theta_{s}) \cdot\\
& \phantom{ \pi_{t}(\theta_{t})} \prod_{s=1}^{t}\frac{1}{M}\sum_{m=1}^{M}\left[f_{s-1}(u_{s}^{m}|\theta_{s-1})\prod_{q\neq m}f_{s}(u_{s}^{m}|\theta_{s-1})\right]
\end{align*}
 where $L_{s}$ has the same r{\^o}le and interpretation as it does
in a standard SMC sampler.

Assume that at iteration $t$ the auxiliary variables $u_{t}^{m}$
are sampled independently (conditional upon the associated value of
the parameter, $\theta_{t-1}$)and identically according to $f_{t}(\cdot|\theta_{t-1})$
and that $K_{t}$ denotes the incremental proposal distribution at
iteration $t$, just as in a standard SMC sampler.

In the absence of resampling, each particle has been sampled from
the following proposal distribution at time $t$: 
\begin{align*}
\widetilde{\mu}_{t}(\widetilde{x}_{t})= & \mu_{0}(\theta_{0})\prod_{s=1}^{t}K_{s}(\theta_{s-1},\theta_{s})\prod_{s=1}^{t}\prod_{m=1}^{M}f_{s}(u_{s}^{m}|\theta_{s-1})
\end{align*}
and hence its importance weight, $W_{t}(\widetilde{x}_{t})$, should be: 
\begin{align*}
& \frac{\pi_{t}(\theta_{t})\prod_{s=0}^{t-1}L_{s}(\theta_{s+1},\theta_{s})}{\mu_{0}(\theta_{0})\prod_{s=1}^{t}K_{s}(\theta_{s-1},\theta_{s})} \\
& \frac{\prod_{s=1}^{t}\frac{1}{M}\sum_{m=1}^{M}\left[f_{s-1}(u_{s}^{m}|\theta_{s-1})\prod_{q\neq m}f_{s}(u_{s}^{m}|\theta_{s-1})\right]}{\prod_{s=1}^{t}\prod_{m=1}^{M}f_{s}(u_{s}^{m}|\theta_{s-1})}\\
= & \frac{\pi_{t}(\theta_{t})\prod_{s=0}^{t-1}L_{s}(\theta_{s+1},\theta_{s})}{\mu_{0}(\theta_{0})\prod_{s=1}^{t}K_{s}(\theta_{s-1},\theta_{s})}\prod_{s=1}^{t}\frac{1}{M}\sum_{m=1}^{M}\frac{f_{s-1}(u_{s}^{m}|\theta_{s-1})}{f_{s}(u_{s}^{m}|\theta_{s-1})}\\
= & W_{t-1}(\widetilde{x}_{t-1})\cdot\frac{\pi_{t}(\theta_{t})L_{t-1}(\theta_{t},\theta_{t-1})}{\pi_{t-1}(\theta_{t-1})K_{t}(\theta_{t-1},\theta_{t})}\\
 & \qquad \frac{1}{M}\sum_{m=1}^{M}\frac{f_{t-1}(u_{t}^{m},\theta_{t-1})}{f_{t}(u_{t}^{m}|\theta_{t-1})},
\end{align*}
which yields the natural sequential importance sampling interpretation.
The validity of the incorporation of resampling follows by standard
arguments.

If one has that $\pi_{t}(\theta_{t})\propto p(\theta_{t})f_{t}(y|\theta_{t})=p(\theta_{t})\gamma_{t}(y|\theta_{t})/Z_{t}(\theta_{t})$
and employs the time reversal of $K_{t}$ for $L_{t-1}$ then one
arrives at an incremental importance weight, at time $t$ of:
\begin{eqnarray*}
\frac{p(\theta_{t})f_{t}(y|\theta_{t-1})}{p(\theta_{t-1})f_{t-1}(y|\theta_{t-1})}\frac{1}{M}\sum_{m=1}^{M}\frac{f_{t-1}(u_{t}^{m}|\theta_{t-1})}{f_{t}(u_{t}^{m}|\theta_{t-1})}=\\
\qquad \frac{p(\theta_{t})\gamma_{t}(y|\theta_{t-1})}{p(\theta_{t-1})\gamma_{t-1}(y|\theta_{t-1})}\frac{1}{M}\sum_{m=1}^{M}\frac{\gamma_{t-1}(u_{t}^{m}|\theta_{t-1})}{\gamma_{t}(u_{t}^{m}|\theta_{t-1})}
\end{eqnarray*}
yielding the algorithm described in section \ref{sub:SMC-with-an}
as an exact SMC algorithm on the described extended space.

\section{Proof of SMC Sampler Error Bound\label{sec:ubound_proof}}

A little notation is required. We allow $(E,\mathcal{E})$ to denote
the common state space of the sampler during each iteration, $\mathcal{C}_{b}(E)$
the collection of continuous, bounded functions from $E$ to $\mathbb{R}$,
and $\mathcal{P}(E)$ the collection of probability measures on this
space. We define the Boltzmann-Gibbs operator associated with a potential
function $G:E\to(0,\infty)$ as a mapping, $\Psi_{G}:\mathcal{P}(E)\rightarrow\mathcal{P}(E)$,
weakly via the integrals of any function $\varphi\in\mathcal{C}_{b}(E)$
\[
\int\varphi(x)\Psi_{G}(\eta)(dx)=\frac{\int\eta(dx)G(x)\varphi(x)}{\int\eta(dx^{\prime})G(x^{\prime})}.
\]

The integral of a set $A$ under a probability measure $\eta$ is
written $\eta(A)$ and the expectation of a function $\varphi$ of
$X\sim\eta$ is written $\eta(\varphi)$. The supremum norm on $\mathcal{C}_{b}(E)$
is defined $||\varphi||_{\infty}=\sup_{x\in E}\varphi(x)$ and the
total variation distance on $\mathcal{P}(E)$ is $||\mu-\nu||_{\text{TV}}=\sup_{A}(\nu(A)-\mu(A))$.
Markov kernels, $M:E\rightarrow\mathcal{P}(E)$ induce two operators,
one on integrable functions and the other on (probability) measures:
\begin{align*}
\forall\varphi\in & \mathcal{C}_{b}(E): & M(\varphi)(\cdot):= & \int M(\cdot,dy)\varphi(y)\\
\forall\mu\in & \mathcal{P}(E): & (\mu M)(\cdot):= & \int\mu(dx)M(x,\cdot).
\end{align*}

Having established this notation, we note that we have the following
recursive definition of the distributions we consider: 
\begin{align*}
\widetilde{\eta}_{0}= & \eta_{0}=:M_{0} & {\eta}_{t\geq1}= & \Psi_{G_{t-1}}(\eta_{t-1}) & \widetilde{\eta}_{t\geq1}= & \Psi_{\widetilde{G}_{t-1}}(\widetilde{\eta}_{t-1})
\end{align*}
and for notational convenience define the transition operators as
\begin{align*}
\Phi_{t}(\eta_{t-1})= & \Psi_{G_{t-1}}(\eta_{t-1})M_{t} & \widetilde{\Phi}_{t}(\widetilde{\eta}_{t-1})= & \Psi_{\widetilde{G}_{t-1}}(\widetilde{\eta}_{t-1})M_{t}.
\end{align*}
We make use of the (nonlinear) dynamic semigroupoid, which we define
recursively, via it's action on a generic probability measure $\eta$,
for $t\in\mathbb{N}$: 
\begin{align*}
\Phi_{t-1,t}(\eta)= & \Phi_{t}(\eta) & \Phi_{s,t}=\Phi_{t}(\Phi_{s,t-1}(\eta))\text{ for }s<t,
\end{align*}
with $\Phi_{t,t}(\eta)= \eta$ and $\widetilde{\Phi}_{s,t}$ defined correspondingly.

We begin with a lemma which allows us to control the discrepancy introduced
by Bayesian updating of a measure with two different likelihood functions.
\begin{lemma}[Approximation Error]\label{lemma:bgerror}
If A1. holds, then $\forall\eta\in\mathcal{P}(E)$ and any $t\in\mathbb{N}$:
\[
||\Psi_{\widetilde{G}_{t}}(\eta)-\Psi_{G_{t}}(\eta)||_{TV}\leq2\gamma.
\]
\end{lemma}
\begin{proof}
Let $\Delta_{t}:=\widetilde{G}_{t}-G_{t}$ and consider a generic
$\varphi\in\mathcal{C}_{b}(E)$:
\begin{align*}
& (\Psi_{\widetilde{G}_{t}}(\eta)-\Psi_{G_{t}}(\eta))(\varphi)\\
= & \frac{\eta(G_{t})\eta(\widetilde{G}_{t}\varphi)-\eta(\widetilde{G}_{t})\eta(G_{t}\varphi)}{\eta(\widetilde{G}_{t})\eta(G_{t})}\\
= & \frac{\eta(G_{t})\eta((G_{t}+\Delta_{t})\varphi)-\eta((G_{t}+\Delta_{t}))\eta(G_{t}\varphi)}{\eta(\widetilde{G}_{t})\eta(G_{t})}\\
= & \frac{\eta(G_{t})\eta(\Delta_{t}\varphi)-\eta(\Delta_{t})\eta(G_{t}\varphi)}{\eta(\widetilde{G}_{t})\eta(G_{t})}
\end{align*}
\par Considering the absolute value of this discrepancy, making using
of the triangle inequality: 
\begin{align*}
  \left|(\Psi_{\widetilde{G}_{t}}(\eta)-\Psi_{G_{t}}(\eta))(\varphi)\right|
\leq & \left\vert \frac{\eta(\Delta_{t}\varphi)}{\eta(\widetilde{G}_{t})}\right\vert +\left\vert \frac{\eta(\Delta_{t})}{\eta(\widetilde{G}_{t})}\right\vert \left\vert \frac{\eta(G_{t}\varphi)}{\eta(G_{t})}\right\vert 
\end{align*}
Noting that $G_{t}$ is strictly positive, we can bound $|\eta(G_{t}\varphi)|/\eta(G_{t})$
with $\eta(G_{t}|\varphi|)/\eta(G_{t})$ and thus with $\left\Vert \varphi\right\Vert _{\infty}$
and apply a similar strategy to the first term: 
\begin{align*}
  & \left|(\Psi_{\widetilde{G}_{t}}(\eta)-\Psi_{G_{t}}(\eta))(\varphi)\right\vert\\
\leq & \left\vert \frac{\eta(|\Delta_{t}|)\left\Vert \varphi\right\Vert _{\infty}}{\eta(\widetilde{G}_{t})}\right\vert +\left\vert \frac{\eta(\Delta_{t})}{\eta(\widetilde{G}_{t})}\right\vert \left\vert \frac{\eta(G_{t}|\varphi|)}{\eta(G_{t})}\right\vert 
\leq 2\gamma\left\Vert \varphi\right\Vert_{\infty}.
\end{align*}
(noting that $\eta(|\Delta_{t}|)/\eta(\widetilde{G_{t}})<\gamma$
by integration of both sides of A1).
\end{proof}
We now demonstrate that, if the local approximation error at each
iteration of the algorithm(characterised by $\gamma$) is sufficiently
small then it does not accumulate unboundedly as the algorithm progresses.


\begin{proof}[Proof of Proposition \ref{th:ubound}] We begin with a telescopic decomposition (mirroring
the strategy employed for analysing particle approximations of these
systems in \citet{DelMoral2004}): 
\begin{align*}
\eta_{t}-\widetilde{\eta}_{t}= & \sum_{s=1}^{t}\Phi_{s-1,t}(\widetilde{\eta}_{s-1})-\Phi_{s,t}(\widetilde{\eta}_{s}).
\end{align*}

We thus establish (noting that $\widetilde{\eta}_{0}=\eta_{0}$):
\begin{align}
\eta_{t}-\widetilde{\eta}_{t}= & \sum_{s=1}^{t}\Phi_{s,t}(\Phi_{s}(\widetilde{\eta}_{s-1}))-\Phi_{s,t}(\widetilde{\Phi}_{s}(\widetilde{\eta}_{s-1})).\label{eq:telescoping_decomp}
\end{align}
\par Turning our attention to an individual term in this expansion,
noting that: 
\begin{align*}
\Phi_{s}(\eta)(\varphi)= & \Psi_{G_{s-1}}(\eta)M_{s}(\varphi) & \widetilde{\Phi}_{s}(\eta)(\varphi)= & \Psi_{\widetilde{G}_{s-1}}(\eta)M_{s}(\varphi)
\end{align*}
we have, by application of a standard Dobrushin contraction argument
and Lemma \ref{lemma:bgerror} 
\begin{align}
& (\Phi_{s}(\widetilde{\eta}_{s-1})-\widetilde{\Phi}_{s}(\widetilde{\eta}_{s-1}))(\varphi)\\
= & \Psi_{G_{s-1}}(\widetilde{\eta}_{s-1})M_{s}(\varphi)-\Psi_{\widetilde{G}_{s-1}}(\widetilde{\eta}_{s-1})M_{s}(\varphi)\nonumber \\
& \left\Vert \Phi_{s}(\widetilde{\eta}_{s-1})-\widetilde{\Phi}_{s}(\widetilde{\eta}_{s-1})\right\Vert _{\mbox{TV}}\\\leq & (1-\epsilon(M))\left\Vert \Psi_{G_{s-1}}(\widetilde{\eta}_{s-1})-\Psi_{\widetilde{G}_{s-1}}(\widetilde{\eta}_{s-1})\right\Vert _{\mbox{TV}}\nonumber \\
\leq & 2\gamma(1-\epsilon(M))\label{eq:local_bound}
\end{align}
which controls the error introduced instantaneously during each step.\par We
now turn our attention to controlling the accumulation of error. We
make use of \citep[Proposition 4.3.6]{DelMoral2004} which, under
assumptions A2 and A3, allows us to deduce that for any probability
measures $\mu,\nu$: 
\begin{align*}
\left\Vert \Phi_{s,s+k}(\mu)-\Phi_{s,s+k}(\nu)\right\Vert _{\mbox{TV}}\leq\beta(\Phi_{s,s+k})\left\Vert \mu-\nu\right\Vert _{\mbox{TV}}
\end{align*}
where 
\begin{align*}
\beta(\Phi_{s,s+k})=\frac{2}{\epsilon(M)\epsilon(G)}(1-\epsilon^{2}(M))^{k}.
\end{align*}
\par Returning to decomposition (\eqref{eq:telescoping_decomp}), applying
the triangle inequality and this result, before finally inserting
(\eqref{eq:local_bound}) we arrive at: 
\begin{align*}
\left\Vert \eta_{t}-\widetilde{\eta}_{t}\right\Vert _{\mbox{TV}}\leq & \sum_{s=1}^{t}\left\Vert \Phi_{s,t}(\Phi_{s}(\widetilde{\eta}_{s-1}))-\Phi_{s,t}(\widetilde{\Phi}_{s}(\widetilde{\eta}_{s-1}))\right\Vert _{\mbox{TV}}\\
\leq & \sum_{s=1}^{t}\frac{2(1-\epsilon^{2}(M))^{t-s}}{\epsilon(M)\epsilon(G)}\left\Vert \Phi_{s}(\widetilde{\eta}_{s-1})-\widetilde{\Phi}_{s}(\widetilde{\eta}_{s-1})\right\Vert _{\mbox{TV}}\\
\leq & \sum_{s=1}^{t}\frac{2(1-\epsilon^{2}(M))^{t-s}}{\epsilon(M)\epsilon(G)}\cdot2\gamma(1-\epsilon(M))\\
= & \frac{4\gamma(1-\epsilon(M))}{\epsilon(M)\epsilon(G)}\sum_{s=1}^{t}(1-\epsilon^{2}(M))^{t-s}
\end{align*}
This is trivially bounded over all $t$ by the geometric series and
a little rearrangement yields the result: 
\begin{align*}
\frac{4\gamma(1-\epsilon(M))}{\epsilon(M)\epsilon(G)}\sum_{s=0}^{\infty}(1-\epsilon^{2}(M))^{s} 
= & \frac{4\gamma(1-\epsilon(M))}{\epsilon^{3}(M)\epsilon(G)}.
\end{align*}
\end{proof}

\section{Pseudo code for random weight SMC sampler \label{sub:smc_pseudo}}

This appendix contains the simplest form of the random weight SMC sampler used in the data point tempering examples in section \ref{sec:Sequential-Monte-Carlo}, in which resampling is performed at every step. Essentially, any standard improvements to SMC algorithms can be applied.

\begin{algorithm}
\caption{Random weight SMC sampler with MCMC move and data point tempering}
\label{smc_mcmc_code}
\begin{algorithmic} 

\For {$p=1$ to $P$}
\State Draw $\theta^{(p)}_{0} \sim p(\cdot)$
\For {$m=1$ to $M$}
\State $u_{1}^{m,p} \sim f_1(\cdot | \theta^{(p)}_{0})$
\EndFor
\State Find the estimate $\widehat{\frac{1}{Z_{1}(\theta_{0}^{(p)})}}$ using \eqref{eq:SMC_Z_tempering_step1}
\State Find incremental weight $\widetilde{w}_{1}^{(p)} = \gamma_1 (y|\theta^{(p)}_0) \widehat{\frac{1}{Z_{1}(\theta_{0}^{(p)})}}$
\EndFor

\State Resample the set of particles and set $w^{(p)}_{t}=1/P$.

\For {$t=1$ to $T$}
\For {$p=1$ to $P$}
\For {$m=1$ to $M$}
\State $u_{t}^{m,p} \sim f_t(\cdot | \theta^{(p)}_{t-1})$
\EndFor
\State Find the estimate $\widehat{\frac{Z_{t-1}(\theta_{t-1}^{(p)})}{Z_{t}(\theta_{t-1}^{(p)})}}$ using \eqref{eq:SMC_Z_tempering-1}
\State Calculate $\widetilde{w}_{t}^{(p)} =  \frac{\gamma_{t}(y|\theta_{t-1}^{(p)})}{\gamma_{t-1}(y|\theta_{t-1}^{(p)})} \widehat{\frac{Z_{t-1}(\theta_{t-1}^{(p)})}{Z_{t}(\theta_{t-1}^{(p)})}}$
\EndFor

\State Resample the set of particles and set $w^{(p)}_{t}=1/P$

\For {$p=1$ to $P$}
\State Draw $\theta^{(p)}_{t} \sim K(\cdot|\theta^{(p)}_{t-1})$ where $K$ is an MCMC kernel
\EndFor
\EndFor
\end{algorithmic}
\end{algorithm}

\end{document}